\documentclass[runningheads,envcountsect]{lmcs}

\usepackage{tikz-cd}
\usepackage[arrow,matrix,frame,curve,cmtip]{xy}\CompileMatrices

\usepackage{stmaryrd}

\usepackage{float}
\usepackage{physics}
\usepackage{amssymb}
\usepackage{enumitem}
\usepackage{mathtools}

\usepackage[utf8]{inputenc}
\usepackage{tikz}
\usetikzlibrary{matrix, arrows.meta, positioning, calc,arrows}
\usepackage[usestackEOL]{stackengine}
\usepackage{proof-at-the-end}
\usepackage[capitalize]{cleveref}

\DeclareFontFamily{U}{mathx}{\hyphenchar\font45}
\DeclareFontShape{U}{mathx}{m}{n}{
	<5> <6> <7> <8> <9> <10>
	<10.95> <12> <14.4> <17.28> <20.74> <24.88>
	mathx10
}{}
\DeclareSymbolFont{mathx}{U}{mathx}{m}{n}
\DeclareFontSubstitution{U}{mathx}{m}{n}
\DeclareMathAccent{\widecheck}{0}{mathx}{"71}
\DeclareMathAccent{\wideparen}{0}{mathx}{"75}

\DeclareFontFamily{OT1}{pzc}{}
\DeclareFontShape{OT1}{pzc}{m}{it}{<-> s * [1.200] pzcmi7t}{}
\DeclareMathAlphabet{\mathpzc}{OT1}{pzc}{m}{it}
\newcommand\functorop[1][l]{\csname#1functor\endcsname}
\newcommand\lfunctorop[3]{%
	\setbox0=\hbox{$#2$}%
	\kern\wd0%
	\ensurestackMath{\Centerstack[c]{#1\\ \mathllap{#2\;\,}\mathclap{\DownArrow}\\#3}}%
}		
\newcommand\rfunctorop[3]{%
	\setbox0=\hbox{$#2$}%
	\ensurestackMath{\Centerstack[c]{#1\\\mathclap{\UpArrow}\mathrlap{\,\;#2}\\#3}}%
	\kern\wd0%
}

\setstackgap{L}{1.3\normalbaselineskip}
\newcommand\UpArrow{\rotatebox[origin=c]{90}{$\longrightarrow$\,}}
\newcommand\DownArrow{\rotatebox[origin=c]{-90}{$\longrightarrow$\,}}
\newcommand\functor[1][l]{\csname#1functor\endcsname}
\newcommand\lfunctor[3]{%
	\setbox0=\hbox{$#2$}%
	\kern\wd0%
	\ensurestackMath{\Centerstack[c]{#1\\ \mathllap{#2\;\,}\mathclap{\DownArrow}\\#3}}%
}
\newcommand\rfunctor[3]{%
	\setbox0=\hbox{$#2$}%
	\ensurestackMath{\Centerstack[c]{#1\\\mathclap{\DownArrow}\mathrlap{\,\;#2}\\#3}}%
	\kern\wd0%
}
\newcommand\functormapsto{\mathrel{\ensurestackMath{\Centerstack[c]{\longmapsto\\ \\\longmapsto}}}}
\setstackgap{L}{1.3\normalbaselineskip}

\newcommand{\hyp}{\catname{Hyp}}
\newcommand{\hyps}{\catname{Hyp}_{\Sigma}}
\newcommand{\gr}{\catname{Graph}}
\newcommand{\dgr}{\catname{SGraph}}
\newcommand{\dg}{\catname{DAG}}
\newcommand{\rt}{\mathsf{dclosed_s}}
\newcommand{\rta}{\mathsf{dclosed}}
\newcommand{\rtd}{\mathsf{dclosed_{d}}}
\newcommand{\catname}[1]{\mathbf{#1}}

\newcommand{\tg}[0]{\catname{TG}_{\Sigma}}
\newcommand{\ari}[0]{\mathsf{ar}}
\newcommand{\lgt}[0]{\mathsf{length}}
\newcommand{\arr}[1]{\mathsf{Mor}(\catname{#1})}
\newcommand{\mono}[1]{\mathsf{Mono}(\catname{#1})}
\newcommand{\reg}[1]{\mathsf{Reg}(\catname{#1})}
\newcommand{\pro}{\mathsf{prod}}
\newcommand{\pred}[1]{{\downarrow}#1}

\newcommand{\id}[1]{\mathsf{id}_{#1}}
\newcommand{\comma}[2]{#1\hspace{1pt} {\downarrow}\hspace{1pt} #2}
\newcommand{\cma}[2]{\mathcal{#1}\hspace{1pt} {\downarrow}\hspace{1pt} \mathcal{#2}}

\newcommand{\tree}{\catname{Tree}}

\newcommand{\renewtheorem}[1]{%
	\expandafter\let\csname #1\endcsname\relax
	\expandafter\let\csname c@#1\endcsname\relax
	\expandafter\let\csname end#1\endcsname\relax
	\newtheorem{#1}%
}

\theoremstyle{plain}

\renewtheorem{thm}{Theorem}[section]
\renewtheorem{cor}[thm]{Corollary}
\renewtheorem{lem}[thm]{Lemma}
\renewtheorem{prop}[thm]{Proposition}

\theoremstyle{definition}

\renewtheorem{rem}[thm]{Remark}
\renewtheorem{exa}[thm]{Example}
\renewtheorem{defi}[thm]{Definition}

\newtheorem*{notaz}{Notation}

\title{A new criterion for $\mathcal{M}, \mathcal{N}$-adhesivity,\\ with an application to hierarchical graphs\rsuper*}
\titlecomment{{\lsuper*}Work supported by the Italian MIUR  project PRIN 2017FTXR7S ``IT-MaTTerS''.
A short version of this paper appears in the \emph{Proceedings of FOSSACS 2022}.
}

\author[D.~Castelnovo]{Davide Castelnovo}	
\address{Department of Mathematics, Computer Science and Physics, University of Udine, Udine, Italy.}	
\email{davide.castelnovo@uniud.it}  

\author[F.~Gadducci]{Fabio Gadducci}	
\address{Department of Computer Science, University of Pisa, Pisa, Italy.}	
\email{fabio.gadducci@unipi.it}  

\author[M.~Miculan]{Marino Miculan}	
\address{Department of Mathematics, Computer Science and Physics, University of Udine, Udine, Italy.}
\email{marino.miculan@uniud.it}	

\begin{document}

\begin{abstract}
\emph{Adhesive categories}, and variants such as $\mathcal{M}, \mathcal{N}$-adhesive ones,  
marked a watershed moment for the algebraic approaches to the rewriting of graph-like structures, providing an abstract framework where many general results (on, e.g., parallelism)  could be recast and uniformly proved.
However, checking that a model satisfies the adhesivity properties is sometimes far from immediate.
In this paper we present a new criterion giving a sufficient condition for \emph{$\mathcal{M}, \mathcal{N}$-adhesivity}, a generalisation of the original notion of adhesivity.
We apply it to several existing categories, and in particular to \emph{hierarchical graphs}, a formalism that is notoriously difficult to fit in the mould of algebraic approaches to rewriting and for which various alternative definitions float around.
\end{abstract}
\maketitle
\section{Introduction}
The introduction of \emph{adhesive categories} marked a watershed moment for the algebraic approaches 
to the rewriting of graph-like structures~\cite{lack2005adhesive,ehrig2006fundamentals}.
Until then, key results of the approaches on e.g. parallelism and confluence had to be proven 
over and over again for each different formalism at hand, 
despite the obvious similarity of the procedure.
Differently from previous solutions to such problems, as the one witnessed by the \emph{butterfly lemma} for 
graph rewriting~\cite[Lemma~3.9.1]{CorradiniMREHL97},
the introduction of adhesive categories provided such a disparate set of formalisms with a common abstract framework 
where many of these general results could be recast and uniformly proved once and for all.
%

\looseness=-1
Despite the elegance and effectiveness of the framework, proving that a given category satisfies the conditions for being adhesive can be a daunting task. 
For this reason, we look for simpler general criteria implying adhesivity for a class of categories.
Similar criteria have been already provided for the core framework of adhesive categories; e.g., every elementary topos is adhesive \cite{lack2006toposes}, and a category is (quasi)adhesive if and only if can be suitably embedded in a topos \cite{johnstone2007quasitoposes,garner2012axioms}. This covers many useful categories such as sets, graphs, and so on. On the other hand, there are many categories of interest which are not (quasi)adhesive, such as directed graphs, posets, and many of their subcategories.  In these cases we can try to prove the more general $\mathcal{M}, \mathcal{N}$-adhesivity for suitable $\mathcal{M}, \mathcal{N}$; however, so far this has been achieved only by means of \emph{ad hoc} arguments. 
To this end, one of the main contributions of this paper is a new criterion for $\mathcal{M}, \mathcal{N}$-adhesivity, based on the verification of some properties of functors connecting the category of interest to a family of suitable adhesive categories. This criterion allows us to prove in a uniform and systematic way some previous results about the adhesivity of categories built by products, exponents, and comma construction.

\looseness=-1
Moreover, it is well-known that categorical properties are often \emph{prescriptive}, indicating abstractly the presence of some good behaviour of the modelled system. Adhesivity is one such property, as it is highly sought after when it comes to rewriting theories. Thus, our criterion for proving $\mathcal{M}, \mathcal{N}$-adhesivity can be seen also as a ``litmus test'' for the given category.
This is useful in situations that are not completely settled, and for which different settings have been proposed.
An important example is that of \emph{hierarchical graphs}, for which we roughly can find two alternative proposals: on the one hand, algebraic formalisms where the edges have some algebraic structures, so that the nesting is a side effect of the term construction; on the other hand, combinatorial approaches where the topology of a standard graph is enriched by some partial order, either on the nodes or on the edges, where the order relation indicates the presence of nesting. 
By applying our criterion, we can show that the latter approach yields indeed an $\mathcal{M}, \mathcal{N}$-adhesive category, confirming and overcoming the limitations of some previous approaches to hierarchical graphs \cite{nyko2012,Padberg17,palacz2004algebraic}, which we briefly recall next.

The more straightforward proposal is by Palacz~\cite{palacz2004algebraic}, using a poset of edges instead of just a set; however, the class of rules has to be restricted 
in order to apply the approach, which in any case predates the introduction of adhesive categories.
Our work allows to rephrase in terms of adhesive properties and generalise Palacz's proposal, dropping the constraint on rules.
Another attempt are Mylonakis and Orejas' \emph{graphs with layers}~\cite{nyko2012}, for which 
$\mathcal{M}$-adhesivity is proved for a class of monomorphisms in the category of symbolic graphs; however, nodes between edges at different layers cannot be shared. 
Padberg \cite{Padberg17} goes  for a coalgebraic presentation via a peculiar ``superpower set'' functor; this gives immediately $\mathcal{M}$-adhesivity provided that this superpower set functor is well-behaved with respect to limits. However, albeit quite general, the approach is rather \emph{ad hoc}, not modular and not very natural for actual modelling.

Summarising, the main contributions of this work are: 
(a) a new general criterion for assessing $\mathcal{M}, \mathcal{N}$-adhesivity; 
(b) new proofs of $\mathcal{M}, \mathcal{N}$-adhesivity for some relevant categories, systematising previous known proofs;
(c) the first proof that a category of hierarchical graph is $\mathcal{M}, \mathcal{N}$-adhesive.
\vfill

\noindent\textit{Synopsis.} 
After having recalled some basic notions, in \cref{sec:criterion} we introduce the new criterion for $\mathcal{M}, \mathcal{N}$-adhesivity; using it, we show $\mathcal{M}, \mathcal{N}$-adhesivity of several constructions, such as products and comma categories.
In \cref{sec:examples} we apply this theory to various example categories, such as directed (acyclic) graphs, trees and term graphs. We show also the adhesivity of several categories obtained by combining adhesive ones, and in particular of the elusive category of hierarchical graphs.
Conclusions and directions for future work are in \cref{sec:concl}. 

\section{$\mathcal{M}, \mathcal{N}$-adhesivity via creation of (co)limits}\label{sec:criterion}
In this section we recall some definitions and results about $\mathcal{M}, \mathcal{N}$-adhesive categories and provide a new criterion to prove this property. 
\vspace{-0.24cm}
\subsection{ $\mathcal{M}, \mathcal{N}$-adhesive categories}
Intuitively, an adhesive category is one in which pushouts of monomorphisms exist and ``behave more or less as they do in the category of sets'' \cite{lack2005adhesive}.  Formally, we require pushouts of monomorphisms to be Van Kampen colimits.
\begin{defi} Given two diagrams
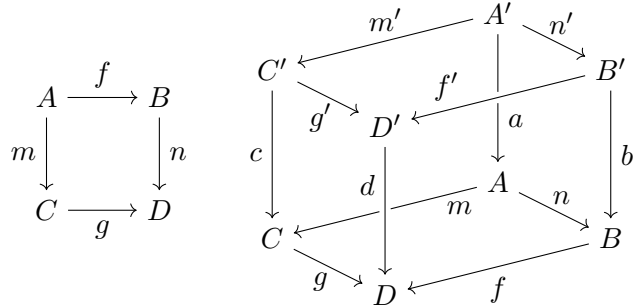
\begin{center}
\begin{tikzpicture}[baseline=(current  bounding  box.center), scale=1.5] 
\node(A) at(-3,1.75){$A$};
\node(B) at (-2,1.75){$B$};
\node(C) at(-2,0.75){$D$};
\node(D) at (-3,0.75){$C$};
\draw[->](A)--(B)node[pos=0.5, above]{$f$};
\draw[->](D)--(C)node[pos=0.5, below]{$g$};
\draw[->](A)--(D)node[pos=0.5, left]{$m$};
\draw[->](B)--(C)node[pos=0.5, right]{$n$};

\node(C)at(-1,0.5){$C$};
\node(A)at(1,1){$A$};
\node(B)at(2,0.5){$B$};
\node(D)at(0,0){$D$};
\node(A')at(1,2.5){$A'$};
\node(B')at(2,2){$B'$};
\node(C')at(-1,2){$C'$};
\node(D')at(0,1.5){$D'$};
\draw[<-](B')--(A')node[pos=0.5, above, xshift=0.1cm, yshift=-0.05cm]{$n'$};
\draw[->](B')--(B)node[pos=0.5, right]{$b$};
\draw[->](C')--(C)node[pos=0.5, left]{$c$}; 
\draw[->](D')--(D)node[pos=0.3, left]{$d$};
\draw[->](A')--(C')node[pos=0.5, above ]{$m'$}; 
\draw[->](B')--(D')node[pos=0.8, above ]{$f'$};
\draw[->](C')--(D')node[pos=0.5, below, xshift=-0.1cm, yshift=0.05cm]{$g'$};
\draw[->](A)--(B)node[pos=0.5, above, xshift=0.1cm, yshift=-0.05cm]{$n$};
\draw[->](C)--(D)node[pos=0.5, below, xshift=-0.1cm, yshift=0.05cm]{$g$};
\draw[->](B)--(D)node[pos=0.5, below]{$f$};
\draw[-](A')--(1,1.8);\draw[<-](A)--(1,1.7)node[pos=0.7, right]{$a$};
\draw[-](A)--(0.05, 0.7625)node[pos=0.2, below]{$m$};
\draw[<-](C)--(-0.05,0.7375);
\end{tikzpicture}
\end{center}
we say that the left square is a \emph{Van Kampen square} if, whenever the right cube has pullbacks as back faces, then its top face is a pushout if and only if the front faces are pullbacks. 
	
	Pushout squares which enjoy the ``if'' of this condition are called \emph{stable}.
\end{defi}
	Given a category $\catname{A}$ we will denote by $\arr{A}, \mono{A}, \reg{A}$ respectively the classes of morphisms, monomorphisms and regular monomorphisms of $\catname{A}$.

\begin{defi}Let $\catname{A}$ be a category and $\mathcal{A}\subseteq \arr{A}$. Then we say that  $\mathcal{A}$ is
\begin{itemize}
	\item 
		\emph{stable under pushouts (pullbacks)} if for every pushout (pullbacks) square 
		
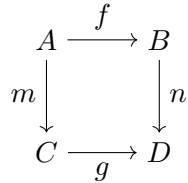
\begin{center}
			\begin{tikzpicture}
			\node(A) at(0,0){$A$};
			\node(B) at (1.5,0){$B$};
			\node(C) at(1.5,-1.5){$D$};
			\node(D) at (0,-1.5){$C$};
			\draw[->](A)--(B)node[pos=0.5, above]{$f$};
			\draw[->](D)--(C)node[pos=0.5, below]{$g$};
			\draw[->](A)--(D)node[pos=0.5, left]{$m$};
			\draw[->](B)--(C)node[pos=0.5, right]{$n$};
			\end{tikzpicture}
		\end{center} if $m \in \mathcal{A}$ ($n\in \mathcal{A}$) then $n \in \mathcal{A}$ ($m \in \mathcal{A}$);
	\item \emph{closed under composition} if $g, f\in \mathcal{A}$ implies $g\circ f\in \mathcal{A}$ whenever $g$ and $f$ are composable;
	\item \emph{closed under $\mathcal{B}$-decomposition} (where $\mathcal{B}$ is another subclass of $\arr{A}$) if $g\circ f\in \mathcal{A}$ and $g\in \mathcal{B}$ implies $f\in \mathcal{A}$;
	\item \emph{closed under decomposition} if it is closed under $\mathcal{A}$-decomposition.
\end{itemize}
\end{defi}
\begin{rem}Clearly, ``decomposition'' corresponds to ``left cancellation'', but we prefer to stick to the name commonly used in literature (see e.g.~\cite{habel2012mathcal}).
\end{rem}

We are now ready to give the definition of $\mathcal{M},\mathcal{N}$-adhesive category \cite{habel2012mathcal,peuser2016composition}.	
\begin{defi}\label{def:class}
	Let $\catname{A}$ be a category and $\mathcal{M}{\subseteq}\mono{A}$, $\mathcal{N}{\subseteq} \arr{A}$ such that
	\begin{enumerate}[label=(\roman*)]
		\item $\mathcal{M}$ and $\mathcal{N}$ contain all isomorphisms and are closed under composition and decomposition;
		\item $\mathcal{N}$ is closed under $\mathcal{M}$-decomposition;
		\item $\mathcal{M}$ and $\mathcal{N}$ are stable under pullbacks and pushouts.
	\end{enumerate}
Then  we say that $\catname{A}$ is \emph{$\mathcal{M}, \mathcal{N}$-adhesive} if
\begin{enumerate}[label=(\alph*)]
 	\item every cospan $C\xrightarrow{g}D\xleftarrow{m}B$ with $m\in \mathcal{M}$ can be completed to a pullback (such pullbacks will be called \emph{$\mathcal{M}$-pullbacks});
	\item every span $C\xleftarrow{m}A\xrightarrow{n}B$ with $m\in \mathcal{M}$ and $n\in \mathcal{N}$ can be completed to a pushout (such pushouts will be called \emph{$\mathcal{M}, \mathcal{N}$-pushouts}); 
	\item  $\mathcal{M}, \mathcal{N}$-pushouts are Van Kampen squares.
\end{enumerate}
\end{defi}

\begin{rem}\label{rem:salva}
 \emph{$\mathcal{M}$-adhesivity} as defined in \cite{azzi2019essence} coincides with $\mathcal{M},\arr{A}$-adhesivity, while 
 \emph{adhesivity} and \emph{quasiadhesivity} \cite{lack2005adhesive,garner2012axioms}  coincide with  $\mono{A}$-adhesivity and $\reg{A}$-adhesivity, respectively. Notice that, in the $\mathcal{M}$-adhesive case, stability under pushouts of $\mathcal{M}$ derives from properties (a)--(c) of \cref{def:class}, while closure under decomposition follows from stability under pullbacks in any category, so there is no need to prove it independently. 
 
Other authors have introduced weaker notions of $\mathcal{M}$-adhesivity; see, e.g.,  \cite{ehrig2006fundamentals,ehrig2004adhesive,sobocinski2020rule}, where our $\mathcal{M}$-adhesive categories are called \emph{adhesive HLR categories}.
\end{rem}

In general, proving that a given category is $\mathcal{M}, \mathcal{N}$-adhesive by verifying  the conditions of \cref{def:class} may be long and tedious; hence, we seek criteria which are sufficient for adhesivity, and simpler to prove.
A prominent example is the following result due to Lack and Soboci{\'n}ski.
\begin{thm}[\cite{lack2006toposes}, Thm. $26$]
	Any elementary topos is an adhesive category.
\end{thm}
In particular the category $\catname{Set}$ of sets and any presheaf category are adhesive.
However, there are many important categories for (graph) rewriting which are not toposes, hence the need for more general criteria. 

We will need some properties of pushouts in the category of sets and functions.

\begin{lem}\label{lem:push} Take a pushout square
	\begin{center}
		\begin{tikzpicture}
		\node(A) at(0,0){$A$};
		\node(B) at (1.5,0){$B$};
		\node(C) at(1.5,-1.5){$D$};
		\node(D) at (0,-1.5){$C$};
		\draw[->](A)--(B)node[pos=0.5, above]{$f$};
		\draw[->](D)--(C)node[pos=0.5, below]{$g$};
		\draw[->](A)--(D)node[pos=0.5, left]{$m$};
		\draw[->](B)--(C)node[pos=0.5, right]{$n$};
		\end{tikzpicture}
	\end{center}
	in $\catname{Set}$, and suppose that $m$ is injective, then
\begin{enumerate}
	\item $n$ is injective too;
	\item the function $B\sqcup C\to D$ induced by $n$ and $g$ is surjective;
	\item for every $x$ and $y\in P$, $x=y$  if and only if one of the following is true:
	\begin{enumerate}
		\item there exists a, necessarily unique, $b\in B$ such that 
		\[x=n(b)=y\]
		\item there exists a unique $c\in C\smallsetminus i(A)$ such that
		\[x=g(c)=y\]	
	\end{enumerate}
\item for every $c_1, c_2\in C$, $g(c_1)=g(c_2)$ if and only if there exists $a_1, a_2$ such that 
\[m(a_1)=c_1 \qquad m(a_2)=c_2 \qquad f(c_1)=f(c_2)\]
\item for every $b\in B$ and $c\in C$, $n(b)=g(c)$ if and only if there exists $a\in A$ such that
\[m(a)=c\qquad f(a)=b\] 
\end{enumerate}

\end{lem}
\begin{proof}The first point follows at once from the adhesivity of $\catname{Set}$, while the others are implied by the explicit description of pushouts in it.
\end{proof}

\subsection{A new criterion for $\mathcal{M}, \mathcal{N}$-adhesivity}
In this section we present our main result, i.e., that $\mathcal{M}, \mathcal{N}$-adhesivity is guaranteed by the existence of a family of functors with sufficiently nice properties. We will adapt some definitions from \cite{adamek2004abstract}.%
\begin{defi} 
	Let $I:\catname{I}\rightarrow \catname{C}$ be a diagram and $J$ a set. We say that a family  $F=\{F_j\}_{j\in J}$ of functors $F_j:\catname{C}\rightarrow \catname{D}_j$
	\begin{enumerate}
		\item \emph{jointly preserves (co)limits} of $I$ if given a (co)limiting (co)cone $(L, l_i)_{i\in\catname{I}}$ for $I$,  every $(F_j(L), F_j(l_i))_{i \in \catname{I}}$ is (co)limiting for $F_j\circ I$;
		\item \emph{jointly reflects (co)limits} of $I$ if a (co)cone $(L, l_i)_{i\in\catname{I}}$ is (co)limiting for $I$ whenever $(F_j(L), F_j(l_i))_{i \in \catname{I}}$ is (co)limiting for $F_j\circ I$ for every $j\in J$;
		\item \emph{jointly lifts (co)limits} of $I$ if given a (co)limiting (co)cone $(L_j, l_{j,i})_{i\in\catname{I}}$ for every $F_j\circ I$, there exists a (co)limiting (co)cone $(L, l_i)_{i\in\catname{I}}$ for $I$ such that $(F_j(L), F_j(l_i))_{i \in \catname{I}}= (L_j, l_{j,i})_{i\in\catname{I}}$ for every $j\in J$;
		\item \emph{jointly creates (co)limits} of $I$ if $I$ has a (co)limit and $F$ jointly preserves and reflects (co)limits along it.
	\end{enumerate}
\end{defi}
\begin{rem}
	Jointly preservation, reflection, lifting or creation of (co)limits of a family $F=\{F_j\}_{j\in J}$ with $F_j:\catname{A}\rightarrow \catname{B}_j$ is equivalent to the usual preservation, reflection, lifting or creation of (co)limits for the functor $\catname{A}\rightarrow \prod_{j\in J}\catname{B}_j$ induced by $F$ (see \cite[Def. V.$1$]{mac2013categories} and \cite[Def. $13.17$]{adamek2004abstract}).
\end{rem} 
\begin{thm}\label{func}
	Let $\catname{A}$ be a category, $\mathcal{M}\subset \mono{A}$, $\mathcal{N}\subset \arr{A}$  satisfying conditions (i)--(iii) of \cref{def:class}, and
	$F$ a  non empty family of functors  $F_j:\catname{A}\rightarrow \catname{B}_j$  such that  $\catname{B}_j$ is $\mathcal{M}_j,\mathcal{N}_j$-adhesive. 
	\begin{enumerate}
		\item If every $F_j$ preserves pullbacks,  $F_j(\mathcal{M})\subset \mathcal{M}_j$ and $F_j(\mathcal{N})\subset \mathcal{N}_j$ for every $j\in J$, $F$ jointly preserves $\mathcal{M}, \mathcal{N}$-pushouts, and jointly reflects pushout squares
		\begin{center}
			\begin{tikzpicture}
			\node(A) at(0,0){$F_j(A)$};
			\node(B) at (2,0){$F_j(B)$};
			\node(C) at(2,-1.5){$F_j(D)$};
			\node(D) at (0,-1.5){$F_j(C)$};
			\draw[->](A)--(B)node[pos=0.5, above]{$F_j(f)$};
			\draw[->](D)--(C)node[pos=0.5, below]{$F_j(g)$};
			\draw[->](A)--(D)node[pos=0.5, left]{$F_j(m)$};
			\draw[->](B)--(C)node[pos=0.5, right]{$F_j(n)$};
			\end{tikzpicture}
		\end{center}
		with $m, n\in \mathcal{M}$ and $f\in \mathcal{N}$, then $\mathcal{M}, \mathcal{N}$-pushouts in $\catname{A}$ are stable. 
		
		Moreover if in addition $F$ jointly reflects $\mathcal{M}$-pullbacks and $\mathcal{N}$-pullbacks then $\mathcal{M}, \mathcal{N}$-pushouts are Van Kampen squares.
		\item If $F$ satisfies the assumptions of the previous points and jointly creates both $\mathcal{M}$-pullbacks and $\mathcal{N}$-pullbacks, then $\catname{A}$ is $\mathcal{M}, \mathcal{N}$-adhesive.
		\item If $F$ jointly creates all pushouts and all pullbacks, then $\catname{A}$ is $\mathcal{M}_F,\mathcal{N}_F$-adhesive, where
		\begin{align*}
			\mathcal{M}_F&:=\{m\in \arr{A}\mid F_j(m)\in \mathcal{M}_j \text{ for every } j\in J\}\\
			\mathcal{N}_F&:=\{n\in \arr{A}\mid F_j(n)\in \mathcal{N}_j \text{ for every } j\in J\}
		\end{align*}
	\end{enumerate}
\end{thm}
\begin{proof}(1.) Take a cube
	in which the bottom face is an $\mathcal{M}, \mathcal{N}$-pushout and all the vertical faces are pullbacks (below, left). Applying any $F_j\in F$ we get another cube
	in $\catname{B}_j$ (below, right) in which the bottom face is an $\mathcal{M}_j, \mathcal{N}_j$-pushout (because $F_j(m)\in \mathcal{M}_j$ and $F_j(n)\in \mathcal{N}_j$) and the vertical faces are pullbacks, thus the top face of the second cube is a pushout for every $j\in J$
		\begin{center}\label{uno}
		\begin{tikzpicture}[baseline=(current  bounding  box.center), scale=1.1] 
		\node(C)at(-1,0.5){$C$};
		\node(A)at(1,1){$A$};
		\node(B)at(2,0.5){$B$};
		\node(D)at(0,0){$D$};
		\node(A')at(1,2.5){$A'$};
		\node(B')at(2,2){$B'$};
		\node(C')at(-1,2){$C'$};
		\node(D')at(0,1.5){$D'$};
		\draw[<-](B')--(A')node[pos=0.5, above, xshift=0.1cm, yshift=-0.05cm]{$n'$};
		\draw[->](B')--(B)node[pos=0.5, right]{$b$};
		\draw[->](C')--(C)node[pos=0.5, left]{$c$}; 
		\draw[->](D')--(D)node[pos=0.3, left]{$d$};
		\draw[->](A')--(C')node[pos=0.5, above ]{$m'$}; 
		\draw[->](B')--(D')node[pos=0.8, above ]{$f'$};
		\draw[->](C')--(D')node[pos=0.5, below, xshift=-0.1cm, yshift=0.05cm]{$g'$};
		\draw[->](A)--(B)node[pos=0.5, above, xshift=0.1cm, yshift=-0.05cm]{$n$};
		\draw[->](C)--(D)node[pos=0.5, below, xshift=-0.1cm, yshift=0.05cm]{$g$};
		\draw[->](B)--(D)node[pos=0.5, below]{$f$};
		\draw[-](A')--(1,1.8);\draw[<-](A)--(1,1.7)node[pos=0.7, right]{$a$};
		\draw[-](A)--(0.05, 0.7625)node[pos=0.2, below]{$m$};
		\draw[<-](C)--(-0.05,0.7375);
		\end{tikzpicture}
\hspace{1cm}
		\begin{tikzpicture}[scale=1.7, baseline=(current  bounding  box.center)] 
		\node(C)at(-1,0.5){$F_j(C)$};
		\node(A)at(1,1){$F_j(A)$};
		\node(B)at(2,0.5){$F_j(B)$};
		\node(D)at(0,0){$F_j(D)$};
		\node(A')at(1,2.5){$F_j(A')$};
		\node(B')at(2,2){$F_j(B')$};
		\node(C')at(-1,2){$F_j(C')$};
		\node(D')at(0,1.5){$F_j(D')$};
		\draw[<-](B')--(A')node[pos=0.5, above, xshift=0.3cm, yshift=-0.07cm]{$F_j(n')$};
		\draw[->](B')--(B)node[pos=0.5, right]{$F_j(b)$};
		\draw[->](C')--(C)node[pos=0.5, left]{$F_j(c)$}; 
		\draw[->](D')--(D)node[pos=0.3, left]{$F_j(d)$};
		\draw[->](A')--(C')node[pos=0.5, above ,yshift=0.05cm]{$F_j(m')$}; 
		\draw[->](B')--(D')node[pos=0.8, above,yshift=0.05cm]{$F_j(f')$};
		\draw[->](C')--(D')node[pos=0.5, below, xshift=-0.25cm, yshift=-0.1cm]{$F_j(g')$};
		\draw[->](A)--(B)node[pos=0.5, above, xshift=0.2cm, yshift=-0.0cm]{$F_j(n)$};
		\draw[->](C)--(D)node[pos=0.5, below, xshift=-0.3cm, yshift=-0.0cm]{$F_j(g)$};
		\draw[->](B)--(D)node[pos=0.5, below, yshift=-0.1cm]{$F_j(f)$};
		\draw[-](A')--(1,1.8);\draw[<-](A)--(1,1.7)node[pos=0.7, right]{$F_j(a)$};
		\draw[-](A)--(0.05, 0.7625)node[pos=0.2, below, yshift=-0.05cm]{$F_j(m)$};
		\draw[<-](C)--(-0.05,0.7375);
		\end{tikzpicture}
	\end{center}
Now $m', f'\in \mathcal{M}$ and $n'\in \mathcal{N}$ since they are the pullbacks of $m$, $f$ and $n$ and thus we can conclude.
	
		Suppose now that $F$ jointly reflects $\mathcal{M}$-pullbacks and $\mathcal{N}$-pullbacks, we have to show that the front faces of the first cube above are pullbacks if the top one is a pushout. In the second cube, the bottom and top face are $\mathcal{M}_j, \mathcal{N}_j$-pushouts and the back faces are pullbacks, then the front faces are pullbacks too by $\mathcal{M}_j, \mathcal{N}_j$-adhesivity. Now, notice that $f\in \mathcal{M}$ and $g\in \mathcal{N}$ (since $\mathcal{M}$ and $\mathcal{N}$ are closed under pushouts) and thus we can conclude since $F$ jointly reflects pullbacks along arrows in $\mathcal{M}$ or in $\mathcal{N}$.

	\noindent
		(2.) Let us show properties (a), (b), (c) defining $\mathcal{M}, \mathcal{N}$-adhesivity.
	\begin{enumerate}[label=(\alph*)]
		\item Given a cospan $C\xrightarrow{g}D\xleftarrow{m}B$ in $\catname{A}$ with $m\in \mathcal{M}$ we can apply $F_j\in F$ to it and get $F_j(C)\xrightarrow{F_j(g)}F_j(D)\xleftarrow{F_j(m)}F_j(B)$ which is a cospan in $\catname{B}_j$  with $F_j(g)\in \mathcal{M}_j$, thus, by hypothesis it has a limiting cone $(P_j, p_{F_j(B)}, p_{F_j(C)})$ in $\catname{B}_j$.	Since $F$ jointly lifts $\mathcal{M}$-pullbacks there exists a limiting cone $(P, p_B, p_C)$ for the cospan  $C\xrightarrow{g}D\xleftarrow{m}B$.
		\item Analogously: for every span $C\xleftarrow{m}A\xrightarrow{n}B$ in $\catname{A}$ with $m\in \mathcal{M}$ and $n\in \mathcal{N}$, we have  $F_j(C)\xleftarrow{F_j(m)}F_j(A)\xrightarrow{F_j(n)}F_j(B)$ in each $\catname{B}_j$ with $F_j(m)\in \mathcal{M}_j$ and $F_j(n)\in \mathcal{N}_j$ and thus there exists a colimiting cocone $(Q_j, q_{F_j(B)}, q_{F_j(C)})$ in $\catname{B}_j$. Now we can conclude because $F$ jointly creates $\mathcal{M}, \mathcal{N}$-pushouts.
		\item This follows at once by the second half of the previous point. 
\end{enumerate}	

\noindent
(3.) By the previous point it is enough to show that $\mathcal{M}_F$ and $\mathcal{N}_F$ satisfy conditions (i)--(iii) of \cref{def:class}.
\begin{enumerate}[label=(\roman*)]
	\item If $f\in \arr{A}$ is an isomorphism then so is $F_j(f)$ for every $F_j\in F$. Thus $F_j(f)$ belongs to $\mathcal{M}_j$ and $\mathcal{N}_j$ for every $j\in J$, implying $f$ is in $\mathcal{M}_F$ and in $\mathcal{N}_F$.  The parts regarding composition and decomposition follow immediately by functoriality of each $F_j\in F$. 
	\item Suppose that $g\circ f\in \mathcal{N}_F$, with $g\in \mathcal{M}_F$ then for every $j\in F$ $F_j(g\circ f)=F_j(g)\circ F_j(f)\in \mathcal{N}_j$ and $F_j(g)\in\mathcal{M}_j$, thus $F_j(f)\in \mathcal{N}_j$ and so $f\in \mathcal{N}_F$.
	\item Take a square 
	\begin{center}
		\begin{tikzpicture}[scale=1.5]
		\node(A) at(0,0){$A$};
		\node(B) at (1,0){$B$};
		\node(C) at(1,-1){$D$};
		\node(D) at (0,-1){$C$};
		\draw[->](A)--(B)node[pos=0.5, above]{$f$};
		\draw[->](D)--(C)node[pos=0.5, below]{$g$};
		\draw[->](A)--(D)node[pos=0.5, left]{$m$};
		\draw[->](B)--(C)node[pos=0.5, right]{$n$};
		\end{tikzpicture}
	\end{center}
and suppose that it is a pullback with $n\in \mathcal{M}_F$ ($\mathcal{N}_F$), then applying any $F_j\in F$ we get that $F_j(m)$ is the pullback of $F_j(n)$ along $F_j(g)$, since $F_j(n)$ is in $\mathcal{M}_j$ (in $\mathcal{N}_j$), which implies that $F_j(m)\in \mathcal{M}_j$ ($\mathcal{N}_j$). This is true for every $j\in J$, from which the thesis follows. Stability under pushouts is proved applying the same argument to $m$. \qedhere
\end{enumerate}
\end{proof}

Applying the previous theorem to the families given by, respectively, projections, evaluations and the inclusion we get immediately the following three corollaries (cfr. also \cite[Thm.~4.15]{ehrig2006fundamentals}). 
\begin{cor}\label{cor:varie1}
	Let $\{\catname{A}\}_{i\in I}$ be a family of categories such that each $\catname{A}_i$ is $\mathcal{M}_i,\mathcal{N}_i$-adhesive. Then the product category $\prod_{i\in I}\catname{A}_i$ is $\prod_{i\in I}\mathcal{M}_i,\prod_{i\in I}\mathcal{N}_i$-adhesive, where
		\begin{align*}
		\prod_{i\in I}\mathcal{M}_i&:=\{(m_i)_{i\in I}\in \mathsf{Mor}(\prod_{i\in I}\catname{A}_i)\mid m_i \in \mathcal{M}_i \text{ for every } i \in I\}\\
		\prod_{i\in I}\mathcal{N}_i&:=\{(n_i)_{i\in I}\in \mathsf{Mor}(\prod_{i\in I}\catname{A}_i)\mid n_i \in \mathcal{N}_i \text{ for every } i \in I\}
		\end{align*}
\end{cor}
\begin{cor}\label{cor:varie2}
 Let $\catname{A}$ be an $\mathcal{M},\mathcal{N}$-adhesive category. Then for every other category $\catname{C}$, the category of functors $\catname{A}^{\catname{C}}$ is $\mathcal{M}^{\catname{C}},\mathcal{N}^{\catname{C}}$-adhesive, where
		\begin{align*}
		\mathcal{M}^{\catname{C}}&:=\{\eta \in \arr{A^C} \mid \eta_C \in \mathcal{M} \text{ for every object } C \text{ of } \catname{C} \}\\
		\mathcal{N}^{\catname{C}}&:=\{\eta \in \arr{A^C} \mid \eta_C \in \mathcal{N} \text{ for every object } C \text{ of } \catname{C} \}
		\end{align*} 
\end{cor}

\begin{cor}\label{cor:varie3}
	Let $\catname{A}$ be a full subcategory of an $\mathcal{M}, \mathcal{N}$-adhesive category $\catname{B}$ and
	$\mathcal{M}'\subset \mono{A}$, $\mathcal{N}'\subset \arr{A}$ satisfying the first three conditions of \cref{def:class} such that $\mathcal{M}'\subset \mathcal{M}$, $\mathcal{N}'\subset \mathcal{N}$ and $\catname{A}$ is closed in $\catname{B}$ under pullbacks and $\mathcal{M'}, \mathcal{N'}$-pushouts. Then $\catname{A}$ is $\mathcal{M}', \mathcal{N'}$-adhesive.
\end{cor}

\subsection{Comma categories}\label{sec:comma}
In this section we will show how to apply \cref{func} to the comma construction \cite{mac2013categories}  in order to guarantee some adhesivity properties under suitable hypotheses.%
\begin{defi}
	For any two functors $L:\catname{A}\rightarrow \catname{C}$, $R:\catname{B}\rightarrow \catname{C}$, the \emph{comma category} $\comma{L}{R}$ is the category in which
	\begin{itemize}
		\item objects are triples $(A, B, f)$ with $A\in \catname{A}$, $B\in \catname{B}$, and $f:L(A)\rightarrow R(B)$; 
		\item a morphism $(A, B, f)\rightarrow (A', B', g)$ is a pair $(h, k)$ with $h:A\rightarrow A'$, $k:B\rightarrow B'$ such that the following diagram commutes
		\begin{center}
			\begin{tikzpicture}
			\node(A) at(0,0){$L(A)$};
			\node(B) at (2,0){$L(A')$};
			\node(C) at(2,-1.5){$R(C')$};
			\node(D) at (0,-1.5){$R(C)$};
			\draw[->](A)--(B)node[pos=0.5, above]{$L(h)$};
			\draw[->](D)--(C)node[pos=0.5, below]{$R(k)$};
			\draw[->](A)--(D)node[pos=0.5, left]{$f$};
			\draw[->](B)--(C)node[pos=0.5, right]{$g$};
			\end{tikzpicture}
		\end{center}
	\end{itemize}
\end{defi} 
We have two obvious forgetful functors 
\begin{equation*}
\begin{split}
U_L:\comma{L}{R}& \rightarrow \catname{A}\\
\functor[l]{(A,B, f)}{(h,k)}{(A', B', g)}
& \functormapsto
\rfunctor{A}{h}{A'}
\end{split}\quad 
\begin{split}
U_R:\comma{L}{R}& \rightarrow \catname{B}\\
\functor[l]{(A,B, f)}{(h,k)}{(A', B', g)}
& \functormapsto
\rfunctor{B}{k}{B'}
\end{split}
\end{equation*}

\begin{exa}$\catname{Graph}$ is equivalent to the comma category made from the identity functor on $\catname{Set}$ and the product functor sending $X$ to $X\times X$.
\end{exa}

We have a classic result relating limits and colimits in the comma category with those preserved by $L$ or $R$.

\begin{lem}\label{colim}
	Let $I:\catname{I}\rightarrow \comma{L}{R}$ be a diagram such that $L$ preserves the colimit (if it exists) of $U_L\circ I$. Then the family $\{U_L, U_R\}$ jointly creates colimits of $I$.
\end{lem}
\begin{proof}
		Suppose that $U_L\circ I$ and $U_R\circ I$ have colimits  $(A, a_i)_{i\in \catname{I}}$ and $(B, b_i)_{i\in \catname{I}}$ respectively, by hypothesis the colimit of $L\circ U_L\circ I$ is $(L(A), L(a_i))_{i\in \catname{I}}$. Now if $I(i)=(A_i, B_i, f_i)$, we have arrows $R(a_i)\circ f_i:L(A_i)\rightarrow R(B)$ that forms a cocone on $L\circ U_L\circ I$: if $l:i\rightarrow j$ is an arrow in $\catname{I}$ then $I(l)=(U_L(I(l), U_R(I(l))$ is an arrow in $\comma{L}{R}$, so
		\begin{align*}
		R(b_j)\circ f_j\circ L(U_L(I(l)))=R(b_j)\circ R(U_R(I(l)))\circ f_i=R(b_j\circ U_R(I(l)))\circ f_i=R(b_i)\circ f_i
		\end{align*}
		thus there exists $f:L(A)\rightarrow R(B)$ such that $f\circ L(a_i)=R(b_i)\circ f_i$. We claim that $(A, B, f)$ with $(a_i, b_i)$ as $i^\text{th}$ coprojection is the colimit of $I$. Let $((X, Y, g), (x_i, y_i))_{i\in \catname{I}}$ be a cocone on $I$, in particular $(X, x_i)_{i\in \catname{I}}$ and $(Y, y_i)_{i\in \catname{I}}$ are cocones on $U_L\circ I$ and $U_R\circ I$ respectively so we have uniquely determined arrows $x:A\rightarrow X$ and $y:B\rightarrow Y$ such that 
		$x\circ a_i=x_i$ and$y\circ b_i=y_i$. We claim that $(x,y)$ is an arrow of $\comma{L}{R}$. For any $i\in \catname{I}$ we have
		\begin{align*}
		&R(y)\circ f\circ L(a_i)=R(y)\circ R(b_i)\circ f_i=R(y\circ b_i)\circ f_i\\=&R(y_i)\circ f_i= g\circ L(x_i)=g\circ L(x\circ a_i)=g\circ L(x)\circ L(a_i)
		\end{align*}
		And since the family $\{L(a_i)\}_{i\in \catname{I}}$ is jointly monic we get that $R(y)\circ f=g\circ L(x)$.		Uniqueness of $(x,y)$ follows at once and so $((A, B, f), (a_i, b_i))_{i \in \catname{I}}$ is a colimit for $I$ which, by construction, is preserved by $U_L$ and $U_R$. Reflection follows by the previous construction: if $((A, B, f), (a_i, b_i))_{i\in \catname{I}}$ is a cone in $\comma{L}{R}$ such that both $(A, a_i)_{i\in \catname{I}}$ and $(B, b_i)_{i\in \catname{I}}$ are colimiting, then the argument above shows that $((A, B, f), (a_i, b_i))_{i\in \catname{I}}$ is colimiting too.
\end{proof} 

Let $L:\catname{A}\to \catname{C}$ and $R:\catname{B}\to \catname{C}$ be two functors with duals $L^{op}:\catname{A}^{op}\to \catname{C}^{op}$ and $R^{op}:\catname{B}^{op}\to \catname{C}^{op}$.  An object $(A, B, f)$ of $\comma{L}{R}$ is just an arrow $f:L(A)\to R(B)$ in $\catname{C}$, must this can be regarded as an arrow $f:R^{op}(B)\to L^{op}(A)$, i.e as an object of $\comma{R^{op}}{L^{op}}$. Moreover, the commutativity in  $\catname{C}$ of the square
\begin{center}
	\begin{tikzpicture}
	\node(A) at(0,0){$L(A)$};
	\node(B) at (2,0){$L(A')$};
	\node(C) at(2,-1.5){$R(C')$};
	\node(D) at (0,-1.5){$R(C)$};
	\draw[->](A)--(B)node[pos=0.5, above]{$L(h)$};
	\draw[->](D)--(C)node[pos=0.5, below]{$R(k)$};
	\draw[->](A)--(D)node[pos=0.5, left]{$f$};
	\draw[->](B)--(C)node[pos=0.5, right]{$g$};
	\end{tikzpicture}
\end{center} is tantamount to the commutativity in $\catname{C}^{op}$ of the square 
\begin{center}
	\begin{tikzpicture}
	\node(A) at(0,0){$R(C')$};
	\node(B) at (2,0){$R(C)$};
	\node(C) at(2,-1.5){$L(A)$};
	\node(D) at (0,-1.5){$L(A')$};
	\draw[->](A)--(B)node[pos=0.5, above]{$R(k)$};
	\draw[->](D)--(C)node[pos=0.5, below]{$L(h)$};
	\draw[->](A)--(D)node[pos=0.5, left]{$g$};
	\draw[->](B)--(C)node[pos=0.5, right]{$f$};
	\end{tikzpicture}
\end{center} 
	Thus we have proved the following.
\begin{prop}\label{prop:dual}
$(\comma{L}{R})^{op}$ is equal to	$\comma{R^{op}}{L^{op}}$, moreover $U^{op}_L=U_{L^{op}}$ and $U^{op}_R=U_{R^{op}}$.
\end{prop}
This easy result allows us to dualize \cref{colim}.
\begin{cor}\label{lim} The family $\{U_L, U_R\}$ jointly creates limits along
	every diagram $I:\catname{I}\rightarrow \comma{L}{R}$ such that $R$ preserves the limit of $U_R\circ I$.
\end{cor}
\begin{proof}
	Apply \cref{prop:dual,colim}.
\end{proof}

Now, in every category an arrow $m:C\to D$ is a mono if and only if the square
\begin{center}
	\begin{tikzpicture}
	\node(A) at(0,0){$C$};
	\node(B) at (1.5,0){$C$};
	\node(C) at(1.5,-1.5){$D$};
	\node(D) at (0,-1.5){$C$};
	\draw[->](A)--(B)node[pos=0.5, above]{$\id{C}$};
	\draw[->](D)--(C)node[pos=0.5, below]{$m$};
	\draw[->](A)--(D)node[pos=0.5, left]{$\id{C}$};
	\draw[->](B)--(C)node[pos=0.5, right]{$m$};
	\end{tikzpicture}
\end{center} 
is a pullback. Thus, using \cref{lim}, we can characterize monos in comma categories.

\begin{cor}\label{cor:mono}
If $R$ preserves pullbacks then an arrow $(h,k)$ in $\comma{L}{R}$ is mono if and only if both $h$ and $k$ are monomorphisms.
\end{cor}
We can also deduce the following result from \cref{func,lim}.
\begin{thm}\label{comma} 
	Let $\catname{A}$ and $\catname{B}$ be respectively $\mathcal{M},\mathcal{N}$-adhesive and $\mathcal{M}',\mathcal{N}'$-adhesive categories,  $L:\catname{A}\rightarrow \catname{C}$ a functor that preserves $\mathcal{M}, \mathcal{N}$-pushouts, and  $R:\catname{B}\rightarrow \catname{C}$ a pullback preserving one. Then $\comma{L}{R}$ is $\cma{M}{M'}, \cma{N}{N'}$-adhesive, where 
	\begin{align*}
	\cma{M}{M}'&:=\{(h,k)\in \mathsf{Mor}(\comma{L}{R}) \mid h\in \mathcal{M}, k\in \mathcal{M}'\}\\
	\cma{N}{N}&:=\{(h,k)\in \mathsf{Mor}(\comma{L}{R}) \mid h\in \mathcal{N}, k\in \mathcal{N}'\}.
	\end{align*}
\end{thm}
When $L=\id{\catname{A}}$ and $R$ is the constant functor into an object $A$, the comma category $\comma{L}{R}$ is just the \emph{slice category} $\catname{A}/A$ over $A$. 
\begin{cor}\label{cor:slice}
	If $A$ is an object of an $\mathcal{M}, \mathcal{N}$-adhesive category $\catname{A}$, then $\catname{A}/A$ is $\mathcal{M}/A, \mathcal{N}/A$-adhesive, where
\begin{align*}\mathcal{M}/A:=\{m\in  \mathsf{Mor}(\catname{A}/A) \mid m\in \mathcal{M} \}\qquad
\mathcal{N}/A:=\{n\in  \mathsf{Mor}(\catname{A}/A) \mid n\in \mathcal{N} \}
\end{align*}	
	
\end{cor}
\paragraph{When is $U_R$ a right adjoint?} We will end this section with a technical result regarding the existence of a left adjoint to $U_R$. This result will be useful to add \emph{interfaces} to various classes of (hyper)graphs (see \cref{subsec:hgraph,subsec:hhgraph}).
\begin{prop}\label{prop:left}
	If $\catname{A}$ has initial objects and $L$ preserves them then the forgetful functor $U_R:\comma{L}{R}\to \catname{B}$ has a left adjoint $\Delta$.
\end{prop}
\begin{proof} For an object $B\in \catname{B}$ define $\Delta(B)$ as $(I, B, !_{B})$, where $I$ is an initial object in $\catname{A}$ and $!_{B}$ is the unique arrow $L(I)\to B$. Let $\id{B}:B\to U_R(\Delta(B))=B$ be the identity, and $k:B\to U_R(A, B', f)$ an arrow in $\catname{B}$. Now, by initiality of $I$, there is only one arrow $h:I\to A $ in $\catname{A}$ and, since $L$ preserves initial objects, the following square commutes.
	\begin{center}
		\begin{tikzpicture}
		\node(A) at(0,0){$L(I)$};
		\node(B) at (2,0){$L(A)$};
		\node(C) at(2,-1.5){$R(B')$};
		\node(D) at (0,-1.5){$R(B)$};
		\draw[->](A)--(B)node[pos=0.5, above]{$L(h)$};
		\draw[->](D)--(C)node[pos=0.5, below]{$R(k)$};
		\draw[->](A)--(D)node[pos=0.5, left]{$!_B$};
		\draw[->](B)--(C)node[pos=0.5, right]{$f$};
		\end{tikzpicture}
	\end{center}
Thus $(h,k)$ is the unique morphism $\Delta(B)\to (A, B', f)$	such that $U_R(h,k)=k\circ \id{B}$.
\end{proof}
Dualizing we get immediately the following.
\begin{cor}If $R$ preserves terminal objects then $U_L:\comma{L}{R}\to \catname{A}$ has a right adjoint.
\end{cor}

\section{Application to some categories of graphs}\label{sec:examples}
In this section we apply the results provided in \cref{sec:criterion}, to some important categories of graphs, such as directed (acyclic) graphs and hierarchical graphs.
These examples have been chosen for their importance in graph rewriting, and because we can recover their $\mathcal{M},\mathcal{N}$-adhesivity in a uniform and systematic way.
In fact, in the case of hierarchical graphs we give the first proof of $\mathcal{M},\mathcal{N}$-adhesivity, to our knowledge.

\subsection{Directed (acyclic) graphs}
Among visual formalisms, directed simple graphs represent one of the most-used paradigms,
since they adhere to the classical view of graphs as relations included in the cartesian product of vertices.
It is also well-known that directed graphs are not quasiadhesive~\cite{johnstone2007quasitoposes}, not even in their acyclic variant. In this section we are going to exploit \cref{cor:varie3} to show that these categories of (acyclic) graphs have nevertheless adhesivity properties.

\begin{defi}
	A \emph{directed graph} $\mathcal{{G}}$ is a  $4$-tuple $(E_\mathcal{{G}}, V_\mathcal{{G}}, s_\mathcal{{G}}, t_\mathcal{{G}})$ where $E_\mathcal{{G}}$ and $V_\mathcal{{G}}$ are sets, called the set of \emph{edges} and \emph{nodes} respectively, and $s_\mathcal{{G}},t_\mathcal{{G}}:E_\mathcal{{G}}\rightrightarrows V_\mathcal{{G}}$ are functions, called \emph{source} and \emph{target}. An edge $e$ is \emph{between} $v$ and $w$ if $s_\mathcal{{G}}(e)=v$ and $t_\mathcal{{G}}(e)=w$, $\mathcal{{G}}(v,w)$ is the set of edges between $v$ and $w$. 
	
	A morphism $\mathcal{{G}}\rightarrow \mathcal{{H}}$ is a pair $(f,g)$ of functions $f:E_\mathcal{{G}}\rightarrow E_\mathcal{H}$, $g:V_\mathcal{{G}}\rightarrow V_\mathcal{H}$ such that the following diagrams commute
	\begin{center}
		\begin{tikzpicture}
		\node(A) at(0,0){$E_\mathcal{{G}}$};
		\node(B) at (1.5,0){$V_\mathcal{{G}}$};
		\node(C) at(1.5,-1.5){$V_\mathcal{{H}}$};
		\node(D) at (0,-1.5){$E_H$};
		\draw[->](A)--(B)node[pos=0.5, above]{$s_\mathcal{{G}}$};
		\draw[->](D)--(C)node[pos=0.5, below]{$s_\mathcal{{H}}$};
		\draw[->](A)--(D)node[pos=0.5, left]{$f$};
		\draw[->](B)--(C)node[pos=0.5, right]{$g$};
		\node(A) at(3,0){$E_\mathcal{{G}}$};
		\node(B) at (4.5,0){$V_\mathcal{{G}}$};
		\node(C) at(4.5,-1.5){$W_\mathcal{{H}}$};
		\node(D) at (3,-1.5){$E_\mathcal{{H}}$};
		\draw[->](A)--(B)node[pos=0.5, above]{$t_\mathcal{{G}}$};
		\draw[->](D)--(C)node[pos=0.5, below]{$t_\mathcal{{H}}$};
		\draw[->](A)--(D)node[pos=0.5, left]{$f$};
		\draw[->](B)--(C)node[pos=0.5, right]{$g$};
		\end{tikzpicture}
	\end{center}	
	We will denote by $\gr$ the category so defined.  A \emph{directed simple graph} is a directed graph in which there is at most one edge between two nodes, $\dgr$ is the full subcategory of $\gr$ given by directed simple graphs.
	
		A \emph{path} $[e_i]_{i=1}^n$ in a directed graph $\mathcal{{G}}$ is a finite and non empty list of edges such that $t_\mathcal{{G}}(e_{i})=s_\mathcal{{G}}(e_{i+1})$ for all $1\leq i\leq {n-1}$. A path is called a \emph{cycle} if $s_\mathcal{{G}}(e_1)=t\mathcal{{G}}_(e_n)$. A \emph{directed acyclic graph} is a directed simple graph without cycles. Directed acyclic graphs form a full subcategory $\dg$ of $\dgr$ and $\gr$.
\end{defi}

\begin{rem}\label{multi} From the definition of $\catname{Graph}$, we can immediately deduce its equivalence to:
	\begin{itemize}
		\item the category $\comma{\id{Set}}{\pro}$, where $\pro$ is the functor $\catname{Set}\to \catname{Set}$ defined as \begin{align*}
		\functor[l]{X}{f}{Y}
		\functormapsto
		\rfunctor{X \times X}{f \times f }{Y \times Y}
		\end{align*} 
		which preserves limits;
		\item  the category of presheaves on $\bullet \rightrightarrows \bullet$,  the category with just two objects and only two parallel arrows between them (besides the identities).
	\end{itemize}
	From these two characterizations we can deduce that $\gr$ is a topos. We can also deduce that limits and colimits of directed graphs are computed component-wise and that an arrow in $\catname{Graph}$ is mono if and only if both its underlying functions are injective.
\end{rem}

We will now establish some properties of $\dgr$ that will be useful in the following.

\begin{prop}\label{mono} If $(f,g):\mathcal{{G}}\rightarrow \mathcal{{H}}$ is an arrow in $\dgr$ with $g$ injective, then $f$ is injective too.
\end{prop}
\begin{proof}Let $e_1, e_2\in E_{\mathcal{{G}}}$ be nodes such that $f(e_1)=f(e_2)$, then
	\begin{align*}g(s_\mathcal{{G}}(e_2))=s_{\mathcal{{H}}}(f(e_2))&=s_\mathcal{{H}}(f(e_1))=g(s_\mathcal{{G}}(e_1))\\g(t_\mathcal{{G}}(e_2))=t_\mathcal{{H}}(f(e_2))&=t_\mathcal{{H}}(f(e_1))=g(t_\mathcal{{G}}(e_1))
	\end{align*}	
Thus
\[s_\mathcal{{G}}(e_1)=s_\mathcal{{G}}(e_2) \qquad t_\mathcal{{G}}(e_1)=t_\mathcal{{G}}(e_2) \]
and we can conclude that $e_1=e_2$ since $\mathcal{{H}}$ is simple.	
\end{proof}
Since $I:\dgr \to \gr$ is full and faithful, then it reflects monomorphisms, thus, from \cref{multi} we get the following.
\begin{cor}\label{mono3}
	An arrow $(f,g):\mathcal{{G}}\to \mathcal{{H}}$ in $\dgr$ is mono if and only if $g$ is injective.
\end{cor}

\begin{defi}Let $\mathcal{{G}}=(E_\mathcal{{G}}, V_\mathcal{{G}}, s_\mathcal{{G}}, t_\mathcal{{G}})$ be a directed graph. We define an equivalence $\sim$ relation on $E_\mathcal{{G}}$ putting 
	\[e_1\sim e_2 \iff s_\mathcal{{G}}(e_1)=s_\mathcal{{G}}(e_2) \text{ and } t_\mathcal{{G}}(e_1)=t_\mathcal{{G}}(e_2)\]
Let $E$ be the quotient $E/\sim$, we define $L(\mathcal{{G}})$ to be the graph $(E,V_\mathcal{{G}}, s,t)$ where $s,t:E\rightrightarrows V_\mathcal{{G}}$ are the functions induced by $s_\mathcal{{G}}$ and $t_\mathcal{{G}}$.
\end{defi}

\begin{rem}By construction, $L(\mathcal{{G}})$ belongs to $\dgr$.
\end{rem}

\begin{prop}\label{prop:fatt}
	The following properties hold
	\begin{enumerate}
		\item the inclusion functor $I:\dgr\rightarrow \gr$ has a left adjoint $L:\gr\rightarrow \dgr$.
		\item an arrow $(f,g):\mathcal{{G}}\rightarrow \mathcal{{H}}$ of $\dgr$ is a regular monomorphism if and only if $f$ is injective and \emph{edge-reflecting}: $\mathcal{{G}}(v_1, v_2)$ is non empty whenever $\mathcal{{H}}(g(v_1), g(v_2))\neq \emptyset$.
	\end{enumerate}
\end{prop}
\begin{proof} 
	\begin{enumerate}
		\item For every object $\mathcal{{G}}$ of $\gr$, there is an arrow $(\pi_\mathcal{{G}}, \id{V_\mathcal{{G}}}):\mathcal{{G}}\rightarrow I(L(\mathcal{{G}}))$. Now, if $\mathcal{{H}}$ is a simple graph and $(f,g):G\to I(\mathcal{{H}})$ a morphism, then $f(e_1)=f(e_2)$ whenever $e_1\sim e_2$, and thus there exists a unique $\qty(\overline{f}, g ):L(\mathcal{{G}})\to \mathcal{{H}}$ such that 
		\[I\qty (\overline{f}, g )\circ (\pi_\mathcal{{G}}, \id{V_\mathcal{{G}}}) =(f,g)\]
		showing that $(\pi_\mathcal{{G}}, \id{V_\mathcal{{G}}})$ is the unit of $L\dashv I$. 
		\item $(\Rightarrow)$. Suppose that $(f,g)$ is the equalizer of $(f_1, g_1), (f_2, g_2): \mathcal{{H}}\rightrightarrows \mathcal{{K}}$, since $I$ preserves limits, $(f,g)$ is the equalizer of $(f_1, g_1)$ and $(f_2, g_2)$ in $\gr$. Let $\mathcal{{G}}'$ be the graph where	
		\begin{equation*}
			E_{\mathcal{{G}}'}:=\{e\in E_H\mid f_1(e)=f_2(e)\}\qquad V_{\mathcal{{G}}'}:=\{v\in V_H\mid v_1(w)=v_2(w)\} 
		\end{equation*}
		and $s_{\mathcal{{G}}'}$, $t_{\mathcal{{G}}'}$ are the restrictions of $s_\mathcal{{H}}$ and $t_\mathcal{{H}}$. Then an equalizer  $(i,j):\mathcal{{G}}'\to \mathcal{{H}}$ of $(f_1, g_1)$ and $(f_2, g_2)$ in $\gr$ is given by the inclusions
		\[i:E_{\mathcal{{G}}'}\to E_\mathcal{{H}} \qquad j:V_{\mathcal{{G}}'}\to V_\mathcal{{H}}\]
		Notice that $\mathcal{{G}}'$ is simple because $\mathcal{{H}}$ is. Now, $I$ preserves limits, so there exists an isomorphism (in $\gr$ and in $\dgr$) $(\phi, \psi):\mathcal{{G}}\to \mathcal{{G}}'$ such that
			\begin{center}
			\begin{tikzpicture}
			\node(A) at(0,0){$\mathcal{{G}}$};
			\node(B) at (3,0){$\mathcal{{H}}$};
			\node(C) at(1.5,-1.5){$\mathcal{{G}}'$};
			\draw[->](A)--(B)node[pos=0.5, above]{$(f,g)$};
			\draw[->](C)--(B)node[pos=0.5, right, yshift=-0.1cm]{$(i,j)$};
			\draw[->](A)--(C)node[pos=0.5, left, yshift=-0.1cm]{$(\phi, \psi)$};
			\end{tikzpicture}
		\end{center}	
		commutes. If we show that $(i,j)$ is edge-reflecting we are done. For every $e\in \mathcal{{H}}(i(v_1), i(v_2))$ then
		\begin{align*}
		s_\mathcal{{K}}(f_1(e))=g_1(s_\mathcal{{H}}(e))=g_1(i(v_1))&=g_2(i(v_1))=g_2(s_\mathcal{{H}}(e))=s_\mathcal{{K}}(f_2(e))\\
		t_\mathcal{{K}}(f_1(e))=g_1(t_\mathcal{{H}}(e))=g_1(i(v_1))&=g_2(i(v_1))=g_2(t_\mathcal{{H}}(e))=t_\mathcal{{K}}(f_2(e))
		\end{align*}
		Thus $f_1(e)=f_2(e)$ because $\mathcal{{K}}$ is simple, i.e. $e\in E_\mathcal{{G}}$.
		
		\medskip 
		\noindent
		$(\Leftarrow)$. Take 
		\[V:=V_\mathcal{{H}}\sqcup (V_\mathcal{{H}}\smallsetminus g(V_\mathcal{{G}}))\]
		and define $E\subseteq V\times V$ putting $(v,v')\in E$ if and only if one of the following is true
		\begin{itemize}
	\item $v=i_1(w)$, $v'=i_1(w')$ and $\mathcal{{H}}(w,w')\neq \emptyset$; 
	\item $v=i_2(w)$, $v'=i_2(w')$ and $\mathcal{{H}}(w,w')\neq \emptyset$; 
	\item 	$v=i_1(w)$, $v'=i_2(w')$ and $\mathcal{{H}}(w,w')\neq \emptyset$; 
	\item $v=i_2(w)$, $v'=i_1(w')$ and $\mathcal{{H}}(w,w')\neq \emptyset$;
	\end{itemize}
	where $i_1$ and $i_2$ are the inclusion of $V_\mathcal{{H}}$ and $V_\mathcal{{H}}\smallsetminus g(V_\mathcal{{G}})$ into $V$. Restricting the projections, we get two arrow $s,t:E\rightrightarrows V$, let $\mathcal{{K}}$ be the directed graph $(E, V, s, t)$, which by construction is simple. 
	
	Now, consider
	\[f:E_\mathcal{{G}}\to V\qquad e \mapsto \qty(i_1\qty(s_\mathcal{{H}}\qty(e)),i_1\qty(t_\mathcal{{H}}\qty(e)) )\]
	paired with $i_1:V_\mathcal{{H}}\to V$ it induces a morphism $(f, i_1):\mathcal{{H}}\to \mathcal{{K}}$. On the other hand, define
	\begin{equation*}
	i':V_\mathcal{{H}}\to V \qquad w \mapsto \begin{cases}
	i_1(w) & w\in g(V_\mathcal{{G}})\\
	i_2(w) & w\notin g(V_\mathcal{{G}})
	\end{cases}
	\end{equation*}
	and 
	\begin{equation*}f':E_\mathcal{{H}}\to E \qquad e \mapsto \begin{cases}
	\qty(i_1\qty(s_\mathcal{{H}}(e)), i_1\qty(t_\mathcal{{H}}(e))) & s_\mathcal{{H}}(e), t_\mathcal{{H}}(e)\in g(V_\mathcal{{G}})\\
	\qty(i_2\qty(s_\mathcal{{H}}(e)), i_2\qty(t_\mathcal{{H}}(e))) & s_\mathcal{{H}}(e), t_\mathcal{{H}}(e)\notin g(V_\mathcal{{G}})\\
	\qty(i_1\qty(s_\mathcal{{H}}(e)), i_2\qty(t_\mathcal{{H}}(e))) & s_\mathcal{{H}}(e)\in g(V_\mathcal{{G}})\\
\qty(i_2\qty(s_\mathcal{{H}}(e)), i_1\qty(t_\mathcal{{H}}(e))) & t_\mathcal{{H}}(e)\in g(V_\mathcal{{G}})
	\end{cases}
	\end{equation*}
	Define now 
		\begin{gather*}
	A:=\qty{e\in E_\mathcal{{H}} \mid s_\mathcal{{H}}(e),t_\mathcal{{H}}(e)\in g(V_\mathcal{{G}})}
	\end{gather*}
	with inclusion $i:A\to E_\mathcal{{H}}$, and let also $j$ be the inclusion $g(V_\mathcal{{H}})\to V_\mathcal{{H}}$. 
	By construction there are arrows $s,t:A\rightrightarrows g(V_\mathcal{{H}})$ such that
	
	\begin{center}
		\begin{tikzpicture}
		\node(A) at(0,0){$A$};
		\node(B) at (1.5,0){$E_\mathcal{{H}}$};
		\node(C) at(1.5,-1.5){$V_\mathcal{{H}}$};
		\node(D) at (0,-1.5){$g(V_\mathcal{{G}})$};
		\draw[->](A)--(B)node[pos=0.5, above]{$i$};
		\draw[->](D)--(C)node[pos=0.5, below]{$j$};
		\draw[->](A)--(D)node[pos=0.5, left]{$s$};
		\draw[->](B)--(C)node[pos=0.5, right]{$s_\mathcal{{H}}$};
		\node(A) at(3,0){$A$};
		\node(B) at (4.5,0){$E_\mathcal{{H}}$};
		\node(C) at(4.5,-1.5){$V_\mathcal{{H}}$};
		\node(D) at (3,-1.5){$g(V_\mathcal{{G}})$};
		\draw[->](A)--(B)node[pos=0.5, above]{$i$};
		\draw[->](D)--(C)node[pos=0.5, below]{$j$};
		\draw[->](A)--(D)node[pos=0.5, left]{$t$};
		\draw[->](B)--(C)node[pos=0.5, right]{$t_\mathcal{{H}}$};
		\end{tikzpicture}
	\end{center}	
	commute. Putting $\mathcal{{G}}':=(A, g(V_\mathcal{{G}}), s, t)$
	we get a (simple) graph, with an inclusion $(i,j):\mathcal{{G}}'\to \mathcal{{G}}$ which is the equalizer in $\gr$ of $(f,i_1)$ and $(f',i')$.
	
Now, $g=j\circ \phi$ for some $\phi: V_\mathcal{{H}}\to g(V_\mathcal{{G}})$ and, since $(f,g)$ is a morphism of $\dgr$,  $f=i\circ \psi$ for some $\psi:E_\mathcal{{H}}\to A$. Notice that
\begin{align*}
j\circ \phi\circ s_\mathcal{{G}}=g\circ s_G&=s_\mathcal{{H}}\circ f=s_\mathcal{{H}}\circ i \circ \psi =j\circ s \circ \psi\\
j\circ \phi\circ t_\mathcal{{G}}=g\circ t_G&=t_\mathcal{{H}}\circ f=t_\mathcal{{H}}\circ i \circ \psi =j\circ t \circ \psi
\end{align*}
Since $j$ is injective, we can deduce that $(\psi, \phi)$ is a morphism $\mathcal{{G}}\to \mathcal{{G}}'$. Now, $\phi$ is surjective by construction and $g$ is injective by hypothesis, thus $\phi$ is injective too and, using \cref{mono3}, we can deduce that also $\psi$ is injective. If we show that $\psi$ is also surjective we are done: let $e\in A$, then $e\in \mathcal{{H}}(g(v_1), g(v_2))$ for some $v_1,v_2\in V_\mathcal{{G}}$, thus there exists $e'\in \mathcal{{G}}(v_1, v_2)$ and, necessarily, $f(e')=e$, but this means that $\psi(e')=e$.
		\qedhere
	\end{enumerate}
\end{proof}

\begin{cor}\label{mono2}
	 The functor $L$ preserves monomorphisms.
\end{cor}
\begin{proof}Let $(f,g):\mathcal{{G}}\to \mathcal{{H}}$ be a monomorphism in $\gr$, then $L(f,g)=(\overline{f}, g)$ where $\overline{f}$ is the unique arrow such that $\overline{f}\circ \pi_\mathcal{{G}}=f$. Now, by \cref{multi} $g$ is injective, thus the thesis follows from \cref{mono3}.
\end{proof}

\begin{exa}
	In \cite{johnstone2007quasitoposes} it is shown that $\dgr$ is not quasiadhesive. Take the cube
		
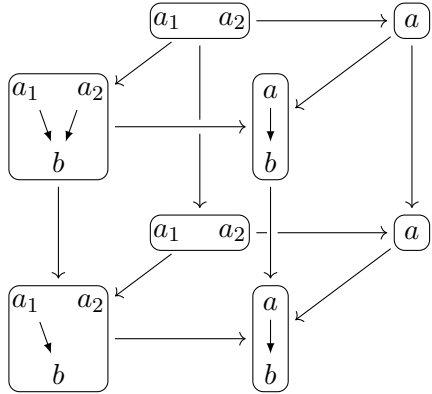
\begin{center}
		\begin{tikzpicture}[scale=0.94]
		\node(A) at(4,0){$a_1 \hspace{0.5cm} a_2$};		
\node(B) at (7,0){$a$};
\node(C) at(5,-1){$a$};
\node(D) at (2,-1){$a_1 \hspace{0.5cm} a_2$};
\node(E) at(2,-2){$b$};
\draw[-latex](D.225)--(E);
\draw[-latex](D.315)--(E);
\node(F) at(5,-2){$b$};
\draw[-latex](C)--(F);

\draw[rounded corners] (6.75, -0.25) rectangle (7.25, 0.25) {};
\draw[rounded corners] (1.3, -2.25) rectangle (2.7, -0.75) {};
\draw[rounded corners] (3.3, -0.25) rectangle (4.7, 0.25) {};
\draw[rounded corners] (4.75, -2.25) rectangle (5.25, -0.75) {};

\draw[rounded corners] (6.75, -3.25) rectangle (7.25, -2.75) {};
\draw[rounded corners] (1.3, -5.25) rectangle (2.7, -3.75) {};
\draw[rounded corners] (3.3, -3.25) rectangle (4.7, -2.75) {};
\draw[rounded corners] (4.75, -5.25) rectangle (5.25, -3.75) {};

\node(G) at(4,-3){$a_1 \hspace{0.5cm} a_2$};		
\node(H) at (7,-3){$a$};
\node(I) at(5,-4){$a$};
\node(L) at (2,-4){$a_1 \hspace{0.5cm} a_2$};
\node(M) at(2,-5){$b$};
\draw[-latex](L.225)--(M);
\node(N) at(5,-5){$b$};
\draw[-latex](I)--(N);
\draw[->](4.8,0)--(6.65,0);
\draw[->](5.1,-3)--(6.65,-3);
\draw[-](4.8,-3)--(4.95,-3);
\draw[->](2,-2.35)--(2,-3.65);
\draw[->](7,-0.35)--(7,-2.65);
\draw[->](5,-2.35)--(5,-3.65);
\draw[->](2.8,-4.5)--(4.65,-4.5);
\draw[->](2.8,-1.5)--(4.65,-1.5);
\draw[->](3.6,-3.3)--(2.8,-3.9);
\draw[->](3.6,-0.3)--(2.8,-0.9);

\draw[->](6.7,-3.225)--(5.35,-4.2375);
\draw[->](6.7,-0.225)--(5.35,-1.2375);

\draw[-](4,-0.35)--(4,-1.4);
\draw[->](4,-1.6)--(4,-2.65);
		\end{tikzpicture}
	\end{center}
	By the results of \cref{prop:fatt} the top and bottom faces are pushouts along regular monos and the back faces are pullbacks, but the front one is not, contradicting the Van Kampen property. The same example shows that even $\dg$ is not quasiadhesive. 
\end{exa}

\begin{defi}
	A monomorphism $(f,g):\mathcal{{G}}\rightarrow \mathcal{{H}}$ in $\gr$ is said to be \emph{downward closed} if, for all $e\in E_\mathcal{{H}}$,  $e\in f(E_\mathcal{{G}})$ whenever $t_\mathcal{{H}}(e)\in g(V_\mathcal{{G}})$. We denote by $\rta$, $\rt$ and $\rtd$ the classes of downward closed morphisms in $\gr$, $\dgr$ and $\dg$ respectively.
\end{defi}

\begin{rem}\label{down}
	The functor $L$ of \cref{prop:fatt} sends downward closed morphisms to downward closed morphisms.
\end{rem}

\begin{rem}By \cref{prop:fatt} it is clear that any downward closed morphism between simple graphs is regular. The vice-versa does not hold: a counterexample is given by  
	\begin{center}
	\begin{tikzpicture}[scale=0.94]
		\node (A) at (0,0){$b$};
		\node (B) at (1.5,0.5){$a$};
		\node (C) at (1.5,-0.5){$b$};
		\draw[rounded corners] (-0.25, -0.25) rectangle (0.25, 0.25) {};
		\draw[rounded corners] (1.25, -0.75) rectangle (1.75, 0.75) {};
		\draw[-latex](B)--(C);
		\draw[->] (0.30, 0)--(1.2,0);	\end{tikzpicture}
	
\end{center}
\end{rem}

\begin{lem}\label{lemma:gpush}$\dgr$ and $\dg$ are closed in $\gr$ under pullbacks. Moreover, $\dgr$ is closed under $\reg{\dgr}, \mono{\dgr}$-pushouts, while $\dg$ is closed under $\rtd, \mono{\dg}$-pushouts.
\end{lem}
\begin{proof} Since 
 $\gr$ is a presheaf category, the pullback of a cospan  $ \mathcal{G}\xrightarrow{(f_1,g_1)}\mathcal{H}\xleftarrow{(f_2,g_2)} \mathcal{K}$
is given $\mathcal{P}$ where
		\begin{equation*}
		\begin{split}
		&V_{\mathcal{P}}=\{(v_1, v_2)\in V_{\mathcal{G}}\times V_{\mathcal{K}}\mid g_1(v_1)=g_2(v_2)\}\\
&s_{\mathcal{P}} : P\rightarrow Q \quad (e_1, e_2) \mapsto (s_{\mathcal{G}}(e_1), s_{\mathcal{K}}(e_2))
		\end{split}
	\quad
	\begin{split}
&E_{\mathcal{P}}=\{(e_1, e_2)\in E_{\mathcal{G}}\times E_{\mathcal{K}}\mid f_1(e_1)=f_2(e_2)\} 
\\
&t_{\mathcal{P}} : P\rightarrow Q \quad (e_1, e_2) \mapsto (t_{\mathcal{G}}(e_1), t_{\mathcal{K}}(e_2))  
	\end{split}
		\end{equation*}
		The two obvious projections give the limiting cone. Now it follows at once that two edges with the same source and target or a cycle in $\mathcal{P}$ would induce parallel edges and cycle in $\mathcal{G}$ and $\mathcal{K}$, thus $\mathcal{P}$ is in $\dgr$ or $\dg$ if $\mathcal{G}$ and $\mathcal{K}$ are in it.

		We are left with pushouts. Let us start again with the presheaf category $\gr$ and a span $ \mathcal{G}\xleftarrow{(f_1,g_1)}\mathcal{H}\xrightarrow{(f_2,g_2)}\mathcal{K}$ in it. Its pushout is given by $\mathcal{P}$ where $E_{\mathcal{P}}$ and $V_{\mathcal{P}}$ are given by
	 the pushouts
	 \begin{center}
	 	\begin{tikzpicture}
	 	\node(A) at(0,0){$E_\mathcal{H}$};
	 	\node(B) at (1.5,0){$E_\mathcal{K}$};
	 	\node(C) at(1.5,-1.5){$E_\mathcal{P}$};
	 	\node(D) at (0,-1.5){$E_\mathcal{G}$};
	 	\draw[->](A)--(B)node[pos=0.5, above]{$f_2$};
	 	\draw[->](D)--(C)node[pos=0.5, below]{$p_2$};
	 	\draw[->](A)--(D)node[pos=0.5, left]{$f_1$};
	 	\draw[->](B)--(C)node[pos=0.5, right]{$p_1$};

	 	\node(A) at(3,0){$V_\mathcal{H}$};
	 	\node(B) at (4.5,0){$V_\mathcal{K}$};
	 	\node(C) at(4.5,-1.5){$V_\mathcal{P}$};
	 	\node(D) at (3,-1.5){$V_\mathcal{G}$};
	 	\draw[->](A)--(B)node[pos=0.5, above]{$ g_2$};
	 	\draw[->](D)--(C)node[pos=0.5, below]{$q_2$};
	 	\draw[->](A)--(D)node[pos=0.5, left]{$g_1$};
	 	\draw[->](B)--(C)node[pos=0.5, right]{$q_1$};
	 	\end{tikzpicture}
	 \end{center}
 $s_\mathcal{P}$ is induced from $q_2\circ s_\mathcal{K}$, $q_1\circ s_\mathcal{G}$ and $t_\mathcal{P}$ from $q_2\circ t_\mathcal{K}$, $q_1\circ t_\mathcal{G}$.
\begin{itemize}
	\item Let $(f_1, g_1)$ and $(f_2, g_2)$ be, respectively, a regular mono and a mono in $\dgr$. Let also $e_1$ and $e_2$ be two elements of $\mathcal{P}(v,v')$. We can use \cref{lem:push} to get the following cases.
	\begin{itemize}
			\item $e_1=p_1(e'_1)$ and $e_2=p_1(e'_2)$ for some $e'_1, e'_2\in e_{\mathcal{K}}$. Then
		\begin{gather*}q_1\qty(s_{\mathcal{K}}\qty(e'_1))=s_{\mathcal{P}}(p_1(e'_1))=s_{\mathcal{P}}(e_1)=v=s_{\mathcal{P}}(e_2)=s_{\mathcal{P}}(p_1(e'_2))=q_1\qty(s_{\mathcal{K}}\qty(e'_2))\\
		q_1\qty(t_{\mathcal{K}}\qty(e'_1))=t_{\mathcal{P}}(p_1(e'_1))=t_{\mathcal{P}}(e_1)=v'=t_{\mathcal{P}}(e_2)=t_{\mathcal{P}}(p_1(e'_2))=q_1\qty(t_{\mathcal{K}}\qty(e'_2))
		\end{gather*}
		But $q_1$ is injective, since $g_1$ is injective and $\catname{Set}$ is adhesive, so 
		\[s_{\mathcal{K}}\qty(e'_1)=s_{\mathcal{K}}\qty(e'_2) \qquad t_{\mathcal{K}}\qty(e'_1)=t_{\mathcal{K}}\qty(e'_2)\], from which we can deduce that $e'_1=e'_2$ and the thesis follows.
		
		\item $e_1=p_2(e'_1)$ and $e_2=p_2(e'_2)$ for some $e'_1, e'_2\in E_{\mathcal{G}}$. Then
		\begin{gather*}q_2\qty(s_{\mathcal{K}}\qty(e'_1))=s_{\mathcal{P}}(p_2(e'_1))=s_{\mathcal{P}}(e_1)=v=s_{\mathcal{P}}(e_2)=s_{\mathcal{P}}(p_2(e'_2))=q_2\qty(s_{\mathcal{K}}\qty(e'_2))\\
		q_2\qty(t_{\mathcal{K}}\qty(e'_1))=t_{\mathcal{P}}(p_2(e'_1))=t_{\mathcal{P}}(e_1)=v'=t_{\mathcal{P}}(e_2)=t_{\mathcal{P}}(p_2(e'_2))=q_2\qty(t_{\mathcal{K}}\qty(e'_2))
		\end{gather*}
		But $g_2$ is injective, so, as before, $q_2$ is injective too and 
		\[s_{\mathcal{G}}(e'_1)=s_{\mathcal{G}}(e'_2) \qquad t_{\mathcal{G}}(e'_1)=t_{\mathcal{G}}(e'_2)\]
		and $e'_1=e'_2$.
		
		\item $e_1=p_1(e'_1)$ and $e_2=p_2(e'_2)$ for some $e'_1\in \mathcal{K}$ and $e'_2\in E_{\mathcal{G}}$. Therefore we have		
		\[p_1\qty(s_{\mathcal{K}}\qty(e'_1))=v=p_2\qty(s_{\mathcal{G}}\qty(e'_2)) \qquad p_1\qty(t_{\mathcal{K}}\qty(e'_1))=v'=p_2\qty(t_{\mathcal{G}}\qty(e'_2)) \]
	Thus there exist $w_1$ and $w_2\in V_{\mathcal{H}}$ such that 
	\[g_1(w_1)=s_{\mathcal{G}}\qty(e'_2), \quad g_2(w_1)=s_{\mathcal{K}}\qty(e'_1) \quad g_1(w_2)=t_{\mathcal{G}}\qty(e'_1), \quad g_2(w_2)=t_{\mathcal{K}}\qty(e'_2) \]
	Thus $e'_1\in \mathcal{G}(g_1(w_1), g_1(w_2))$, but $(f_1, g_1)$ is regular, so \cref{prop:fatt} entails the existence of $e\in\mathcal{H}(w_1, w_2)$. Now, $f_1(e)=e'_1$, while 
	\[s_{\mathcal{K}}(f_2(e))=g_2(s_{\mathcal{H}}(e))=g_2(w_1)=s_{\mathcal{K}}(e'_1) \quad t_{\mathcal{K}}(f_2(e))=g_2(t_{\mathcal{H}}(e))=g_2(w_1)=t_{\mathcal{K}}(e'_1)\]
	and thus $f_2(e)=e'_1$. We conclude that $e_1=e_2$ in $E_\mathcal{P}$
	
	\item $e_1=p_2(e'_1)$ and $e_2=p_1(e'_2)$ for some $e'_1\in \mathcal{G}$ and $e'_2\in E_{\mathcal{K}}$. This is done exactly as in the previous point swapping the roles of $e'_1$ and $e'_2$.
\end{itemize} 
\item Let $(f_1, g_1)$ and $(f_2,g_2)$ be, respectively, a downward closed morphism and a mono in $\dg$. Suppose that a cycle $[e_i]_{i=1}^n$ in $\mathcal{P}$ is given. We split again the cases using \cref{prop:push}.
\begin{itemize}
	\item For every $1\leq i\leq n$, $e_i=p_1(e'_i)$ for $e'_i\in E_{\mathcal{K}}$. Then 
	\begin{gather*}q_1\qty(s_{\mathcal{K}}\qty(e'_1))=s_{\mathcal{P}}\qty(e_1)=t_{\mathcal{P}}\qty(e_n)=q_1\qty(t_{\mathcal{K}}\qty(e'_n))\\
q_1\qty(t_{\mathcal{K}}\qty(e'_i))=t_{\mathcal{P}}\qty(e_i)=	s_{\mathcal{P}}(e_{i+1})=q_1\qty(t_{\mathcal{K}}\qty(e'_{i+1}))
	\end{gather*}
As before, $q_1$ is injective because is the pushout of an injective functions, thus $[e'_i]_{i=1}^n$	is a cycle in $\mathcal{K}$, which is absurd.
	\item For every $1\leq i\leq n$, $e_i=p_2(e'_i)$ for $e'_i\in E_{\mathcal{G}}$. Then 
	\begin{gather*}q_2\qty(s_{\mathcal{G}}\qty(e'_1))=s_{\mathcal{P}}\qty(e_1)=t_{\mathcal{P}}\qty(e_n)=q_2\qty(t_{\mathcal{G}}\qty(e'_n))\\
	q_2\qty(t_{\mathcal{G}}\qty(e'_i))=t_{\mathcal{P}}\qty(e_i)=	s_{\mathcal{P}}(e_{i+1})=q_2\qty(t_{\mathcal{G}}\qty(e'_{i+1}))
	\end{gather*}
	Even in this case we can conclude appealing to the injectivity of $q_2$.
\end{itemize} 
\medskip 

To deal with the other cases we can reason in the following way. Take $e= p_1(e')$ for some $e'\in E_{\mathcal{K}}$ and suppose that there exists $a=p_2(a')$ for some $a'\in E_{\mathcal{G}}$ such that $s_{\mathcal{P}}(e)=t_{\mathcal{P}}(a)$.
By \cref{lem:push} there exists  $v\in V_{\mathcal{H}}$ such that
\[q_2\qty(g_1(v))=t_{\mathcal{P}}(a)=q_2\qty(p_2(a'))\]
$q_2$ is injective, thus $g_1(v)=p_2(a')$. Since $(f_1, g_1)\in \rtd$ there exists $b\in E_{\mathcal{H}}$ such that $f_1(b)=a'$. Thus $a=p_1(f_2(b))$ belongs to $p_1$.

Let us apply this argument to our cycle $[e_i]^n_{i=1}$. By \cref{lem:push} and the second point above, there must be an index $j$ such that $e_j\in p_1(E_{\mathcal{K}})$. Now, if $j>1$ the previous argument shows that $e_{j-1}\in p_1(E_{\mathcal{K}})$ too, thus surely $e_1\in  p_1(E_{\mathcal{K}})$. But, since  $[e_i]^n_{i=1}$ is a cycle, the same argument shows that $e_n\in  p_1(E_{\mathcal{K}})$ and this implies that every $e_i\in e_1\in  p_1(E_{\mathcal{K}})$ for every $1\leq i \leq n$, but we already know that this is absurd.
	 \qedhere
\end{itemize}		
\end{proof}

\begin{thm}The category $\dgr$  is  both $\reg{\dgr}, \mono{\dgr}$- adhesive and $\mono{\dgr},\reg{\dgr}$-adhesive, while $\dg$ is $\rtd, \mono{\dg}$-adhesive.
\end{thm}
\begin{proof}
	In light of \cref{func} we only have to show that the right classes of arrows satisfies the properties of \cref{def:class}.
	Clearly all classes contains all isomorphisms and are closed under composition. $\mono{A}$ is closed under decomposition,  and $\reg{A}$ is closed under $\mono{A}$-decomposition for every category $\catname{A}$, so $\mono{\dgr}, \reg{\dgr}$ and $\reg{\dg}$ are closed under decomposition, $\reg{\dgr}$
 under $\mono{\dgr}$-decomposition, $\mono{\dgr}$ under $\reg{\dgr}$-decomposition and, finally, the class $\mono{\dg}$ under $\rtd$-decomposition.  Moreover they are all closed under pullbacks. Now for the other properties
	\begin{itemize}
		\item $\mono{\dg}$ and $\reg{\dg}$ are closed under pushout. The first one follows from \cref{mono2} and the adhesivity of $\gr$, while the second one follows from the explicit construction of pushouts in $\gr$ and in $\dgr$.
		\item $\rtd$ is stable under pullbacks. 
		It follows from the explicit construction of pullbacks. Indeed consider a pullback square with $(f,g)\in \rtd$
		\begin{center}
			\begin{tikzpicture}[scale=1]
			\node(A) at(0,0){$\mathcal{{P}}$};
			\node(B) at (1.5,0){$\mathcal{K}$};
			\node(C) at(1.5,-1.5){$\mathcal{H}$};
			\node(D) at (0,-1.5){$\mathcal{G}$};
			\draw[->](A)--(B)node[pos=0.5, above]{$(\pi_2, \pi'_2)$};
			\draw[->](D)--(C)node[pos=0.5, below]{$(h,k)$};
			\draw[->](A)--(D)node[pos=0.5, left]{$(\pi_1, \pi'_1)$};
			\draw[->](B)--(C)node[pos=0.5, right]{$(f,g)$};
			\end{tikzpicture}
		\end{center}
		Take $v_1\in \pi_{1}(V_{\mathcal{P}})$, then there exists $v_2\in V_{\mathcal{K}}$ such that $g(v_2)=k(v_1)$. Now, let $e_1\in E_\mathcal{G}$ such that $t_\mathcal{G}(e_1)=v_1$, we have
		\[t_{\mathcal{H}}(h(e_1))=k(t_{\mathcal{G}}(e_1))=k(v_1)=g(v_2)\]
		  but $(f,g)\in \rtd$, and so there exist $w_2\in V_{\mathcal{K}}$ and $e_2\in E_{\mathcal{K}}$ such that 
		  \[f(e_2)=h(e_1)\qquad g(w_2)=s_{\mathcal{H}}(h(e_1))=k(s_{\mathcal{G}}(e_1))\]
		  
		   hence $(s_{\mathcal{G}}(e_1), w_2)\in V_{\mathcal{P}}$ and $(e_1, e_2)\in E_{\mathcal{P}}$ and this means that $s_{\mathcal{G}}(e_1)\in \pi_1(V_{\mathcal{P}})$ and $e_1\in \pi'_1(E_\mathcal{P})$, i.e. that $(\pi_1, \pi'_1)\in \rtd$.
		\item $\rtd$ is stable under pushouts. 
		Take a pushout square in $\gr$
		\begin{center}
			\begin{tikzpicture}[scale=1]
			\node(A) at(0,0){$\mathcal{{H}}$};
			\node(B) at (1.5,0){$\mathcal{K}$};
			\node(C) at(1.5,-1.5){$\mathcal{P}$};
			\node(D) at (0,-1.5){$\mathcal{G}$};
			\draw[->](A)--(B)node[pos=0.5, above]{$(f_2, g_2)$};
			\draw[->](D)--(C)node[pos=0.5, below]{$(p_2,q_2)$};
			\draw[->](A)--(D)node[pos=0.5, left]{$(f_1, g_1)$};
			\draw[->](B)--(C)node[pos=0.5, right]{$(p_1,q_1)$};
			\end{tikzpicture}
		\end{center}
		we know by the proof of \cref{lemma:gpush}  that its pushout $\mathcal{P}$ does not contains cycles (even if can contain parallel edges). Applying $L$ to it we get a pushout in $\dgr$ that is acyclic, therefore, since $\dg$ is a full subcategory of $\dgr$, a pushout in $\dg$. So, by \cref{down} it is enough to show that $\rta$ is closed under pushouts. But this now follows by the description of pushouts in $\gr$.
		
		Indeed let $e\in E_{\mathcal{P}}$ such that $t_{\mathcal{P}}(e)=q_1(v)$ for some $v\in V_{\mathcal{K}}$. Suppose that $e\notin p_1(E_\mathcal{K})$, by \cref{lem:push} we know that there exists $e'\in E_{\mathcal{G}}$ such that $p_2(e')=e$, but then 
		\[q_1(v)=t_{\mathcal{P}}(e)=q_2\qty(t_{\mathcal{G}}\qty(e'))\]
		Thus there exists $w\in V_{\mathcal{H}}$ such that \[g_1(w)=t_{\mathcal{G}}\qty(e') \qquad g_2(w)=v\]		
	 Since, by hypothesis, $(f_1,g_1)$ is in $\rta$, there exists $e''\in E_{\mathcal{H}}$ such that $f_1(e'')=e'$, thus
	 \[p_1\qty(f_2\qty(e''))=p_2\qty(f_1\qty(e''))=p_2(e')=e\]
	 and $e\in p_1(E_{\mathcal{K}})$. \qedhere 
	\end{itemize}
\end{proof} 

\subsection{Tree Orders}
In this section we present \emph{trees} as partial orders and show that the resulting category is actually a topos of presheaves, hence adhesive.
This fact will be exploited in \cref{subsec:hhgraph} to construct a category of hierarchical graphs, where the hierarchy between edges is modelled by trees. 
\begin{defi}
	A \emph{tree order} is a partial order $(E, \leq)$ such that for every $e\in E$, $\pred{e}$ is a finite set totally ordered by the restriction of $\leq$. Since $\pred{e}$ is a finite chain we can define the \emph{immediate predecessor function}
	\begin{equation*}
		i_E:E\rightarrow E\sqcup \{\ast\}\qquad 
		e \mapsto \begin{cases}
			\max(\pred{e}\smallsetminus\{e\}) & \pred{e}\neq \{e\}\\
			\ast &\pred{e}=\{e\}
		\end{cases}
	\end{equation*}
	For any $k\in \mathbb{N}_+$ we can define the \emph{$k^\text{th}$ predecessor function} $i^k_E:E\rightarrow E\sqcup \{\ast\}$ by induction
	\begin{equation*}
		e \mapsto \begin{cases}
			i_E(i_E^{k-1}(e)) & i_E^{k-1}(e)\in E\\
			\ast & i_E^{k-1}(e)=\ast 
		\end{cases}
	\end{equation*}
	in which we take $i^0_E$ to be the inclusion $E\rightarrow E\sqcup\{*\}$.
	
	\noindent
	Let $f:(E, \leq)\rightarrow (F, \leq)$ be a monotone map and $f_\ast:E\sqcup \{\ast\}\rightarrow F\sqcup \{\ast\}$ be its extension sending $\ast$ to  $\ast $. 
	We say that $f$ is \emph{strict} if the following diagram commutes
	\begin{center}
		\begin{tikzpicture}
			\node(A) at(0,0){$E$};
			\node(B) at (2,0){$E\sqcup\{\ast\}$};
			\node(C) at(2,-1.5){$F\sqcup\{\ast\}$};
			\node(D) at (0,-1.5){$F$};
			\draw[->](A)--(B)node[pos=0.5, above]{$i_E$};
			\draw[->](D)--(C)node[pos=0.5, above]{$i_F$};
			\draw[->](A)--(D)node[pos=0.5, left]{$f$};
			\draw[->](B)--(C)node[pos=0.5, right]{$f_\ast$};
		\end{tikzpicture}
	\end{center}
	We define $\tree$ as the subcategory of $\catname{Poset}$ given by tree orders and strict morphisms.
\end{defi}

\begin{exa}A strict morphisms is simply a monotone function that preserves immediate predecessors (and thus every predecessor). For instance the function $\{0\}\rightarrow \{0,1\}$ sending $0$ to $1$ and where we endow the codomain with the order $0\leq 1$,  is not a strict morphism.
\end{exa}

\begin{rem}
	Clearly $i^1_E=i_E$ and it holds that $i^{k}_E(e)=\ast$ if and only if $\abs{\pred{e}}\leq k$. In this case an easy induction shows that
	$	\abs{\pred{i^{k}_E(e)}}=\abs{\pred{e}}-k$.
\end{rem}

\begin{rem}\label{forget} We have an obvious forgetful functor
	\begin{align*}\abs{-}:\tree& \rightarrow \catname{Set}\\
		\functor[l]{(E, \leq)}{f}{(F,\leq)}
		& \functormapsto
		\rfunctor{E}{f}{F}
	\end{align*}	
\end{rem}
\begin{rem}
	Let $(E, \leq)$ be an object of $\tree$ and $\omega$ the first infinite ordinal, then we can define its \emph{associated presheaf} $\widehat{E}:\omega^{op}\rightarrow \catname{Set}$ sending $n$ to the set
	\begin{equation*}
	\{e \in E \mid \abs{\pred{e}\smallsetminus\{e\}}=n\}
\end{equation*} 
If $n\leq m$ in $\omega$, we can define a function
\begin{equation*}
\iota^E_{n,m}:\widehat{E}(m)\rightarrow \widehat{E}(n)\qquad
e \mapsto i_E^{m-n}(e)
\end{equation*}
which is well defined since $\abs{\pred{e}}> m-n$ so 
\begin{align*}
\abs{\pred{i_E^{m-n}(e)}}=\abs{\pred{e}}-m+n=m+1-m+n=n+1
\end{align*} Notice that if $m=n$, $i_E^{m-n}(e)$ is the identity, while for any $k\leq n \leq m$ we have
\begin{align*}
\iota^E_{k,n}(\iota^E_{n,m}(e))=i_E^{n-k}(i_E^{m-n}(e))=i_E^{n-k+m-n}(e)=i_E^{m-k}(e)=\iota^E_{m-k}(e)
\end{align*}
so $\widehat{E}$ is really a presheaf on $\omega$.
\end{rem}
\begin{thm} There exists an equivalence of categories $\widehat{(-)}:\tree\rightarrow \catname{Set}^{\omega^{op}}$ sending $(E, \leq)$ to $\widehat{E}$.
\end{thm}
\begin{proof}
		Let $f:(E, \leq)\rightarrow (F,\leq)$ be an arrow in $\tree$, then an easy induction shows that it must send $e\in \widehat{E}(n)$ in $\widehat{F}(n)$
		\begin{itemize}
			\item if $n=0$ then $		i_{F}(f(e))=f_*(i_E(e))=*$, so
			so $\pred{f(e)}=\emptyset$ and thus $f(e)\in \widehat{F}(0)$;
			\item if $n\geq 1$ since $e\in \widehat{E}(n)$, then $i_E(e)\in \widehat{E}(n-1)$ and, by the inductive hypothesis $f(i_E(e))\in \widehat{F}(n-1)$ and $		f(i_E(e))=f_*(i_E(e))=i_F(f(e))$,
			so $i_F(f(e))\in \widehat{F}(n-1)$ and thus $f(e)\in \widehat{F}(n)$.
		\end{itemize}
		Therefore we can define 
		\begin{align*}
			\widehat{f}_{n}:\widehat{F}(n)\rightarrow \widehat{G}(n)\qquad
			e\mapsto f(e)
		\end{align*}
		and, for every $n\leq m$ and $e\in \widehat{E}(m)$ we have
		\begin{align*} 
			\widehat{f}_n({\iota^E_{n,m}}(e))=f(i^{m-n}_E(e))=f_*(i^{m-n}_E(e))=i^{m-n}_F(f(e))=\iota^F_{n,m}(\widehat{f}_n(e))
		\end{align*}
		where the middle step follows easily by induction from the definition of strict morphism. Thus we can define the functor $\widehat{(-)}:\tree \rightarrow \catname{Set}^{\omega^{op}}$, we want to show that it is an equivalence. It is clearly faithful while, for every $\eta:\widehat{E}\rightarrow \widehat{F}$, we can define
		\begin{align*}
			\overline{\eta}:(E, \leq ) \rightarrow (F, \leq )\qquad 
			e \mapsto \eta_{\abs{\pred{e}}-1}(e)
		\end{align*}
		that is easily seen to be strict. This prove fullness. For essential surjectivity: given $F:\omega^{op}\rightarrow \catname{Set}$ we define  $\widecheck{F}$ as the poset in which
		\begin{itemize}
			\item the underlying set is given by $\sqcup_{n\in \omega}F(n)$;
			\item  $x\leq y$ if and only if $x=F(l_{n,m})(y)$ where $x\in F(n)$, $y\in f(m)$ and $l_{n,m}$ is the arrow corresponding to $n\leq m$.
		\end{itemize} 
		For every $e\in \sqcup_{n\in \omega}F(n) $ it holds that
		\begin{equation*}
			\pred{e}=\{x\in \sqcup_{n\in \omega}F(n)\mid x= F(l_{m,n})(e) \text{ for some } m \leq n \}
		\end{equation*}
		and so $\abs{\pred{x}}\leq n+1$. On the other hand if $x=F(l_{m,n})(e)$ and $y=F(l_{k,n})(e)$ with, say, $m\leq k$, then
		\begin{align*}
			x=F(l_{m,n})(e)=F(l_{n,k}(l_{m,k}(e)))=F(l_{n,k}(y))
		\end{align*}
		thus $x\leq y$ and $\widecheck{F}\in \tree$.
\end{proof}

\begin{cor}\label{push2}
	$\tree$ is adhesive and the forgetful functor $\abs{-}:\tree \rightarrow \catname{Set}$ preserves all colimits.
\end{cor}	
\begin{proof}
	Let $\widehat{(-)}$ be the equivalence  constructed in the previous theorem, and define $\sqcup:\catname{Set}^{\omega^{op}} \rightarrow \catname{Set}$ as
	\begin{align*}
	\functor[l]{F}{\eta}{G}
 \functormapsto
	\rfunctor{\sqcup_{n\in \omega}F(n)}{\sqcup_{n\in \omega}\eta_n}{\sqcup_{n\in \omega}G(n)}
	\end{align*}	
	since colimits are computed component-wise in $\catname{Set}^{\omega^{op}}$ and coproducts in $\catname{Set}$ commute with colimits we get that $\sqcup$ preserves them. Now it is enough to notice that the following triangle commutes
\begin{equation*}
	\begin{tikzpicture}
	\node(A) at(0,0){$\tree$};
	\node(B) at (3,0){$\catname{Set}^{\omega^{op}}$};
	\node(C) at(1.5,-1.5){$\catname{Set}$};
	\draw[->](A)--(B)node[pos=0.5, above]{$\widehat{(-)}$};
	\draw[->](A)--(C)node[pos=0.5, left,xshift=-0.1cm, yshift=-0.1cm]{$\abs{-}$};
	\draw[->](B)--(C)node[pos=0.5, right, xshift=0.1cm,yshift=-0.1cm]{$\sqcup$};
	\end{tikzpicture}
\qedhere 
\end{equation*}
\end{proof}

\subsection{Hierarchical graphs}\label{subsec:hgraph}
We can use trees to produce a category of hierarchical graphs \cite{palacz2004algebraic}, which, in addition, can be equipped with an interface, modelled by a function into the set of nodes. Let us start with graphs.

\begin{defi} A \emph{hierarchical graphs} $\mathcal{G}$ is a  $4$-tuple $((E_\mathcal{G}, \leq), V_\mathcal{G}, s_\mathcal{G}, t_\mathcal{G})$ made by a tree order $(E_\mathcal{G}, \leq)$, a set $V_\mathcal{G}$ and functions $s_\mathcal{G},t_\mathcal{G}:E_\mathcal{G}\rightrightarrows V_\mathcal{G}$. A \emph{morphism} $\mathcal{G}\rightarrow \mathcal{H}$ is a pair $(f,g)$  with $f:(E, \leq)\rightarrow (F, \leq)$ in $\tree$ and  $g:V_\mathcal{G}\rightarrow V_\mathcal{H}$ in $\catname{Set}$ such that the following squares commute
	\begin{center}
		\begin{tikzpicture}
		\node(A) at(0,0){$E_\mathcal{G}$};
		\node(B) at (1.5,0){$V_\mathcal{G}$};
		\node(C) at(1.5,-1.5){$V_\mathcal{H}$};
		\node(D) at (0,-1.5){$E_\mathcal{G}$};
		\draw[->](A)--(B)node[pos=0.5, above]{$s_\mathcal{G}$};
		\draw[->](D)--(C)node[pos=0.5, below]{$s_\mathcal{H}$};
		\draw[->](A)--(D)node[pos=0.5, left]{$\abs{f}$};
		\draw[->](B)--(C)node[pos=0.5, right]{$g$};
		\node(A) at(3,0){$E_\mathcal{G}$};
		\node(B) at (4.5,0){$V_\mathcal{G}$};
		\node(C) at(4.5,-1.5){$V_\mathcal{H}$};
		\node(D) at (3,-1.5){$E_\mathcal{H}$};
		\draw[->](A)--(B)node[pos=0.5, above]{$t_\mathcal{G}$};
		\draw[->](D)--(C)node[pos=0.5, below]{$t_\mathcal{H}$};
		\draw[->](A)--(D)node[pos=0.5, left]{$\abs{f}$};
		\draw[->](B)--(C)node[pos=0.5, right]{$g$};
		\end{tikzpicture}
	\end{center}
This data, with componentwise composition, form a category $\catname{HGraph}$.
\end{defi}
$\catname{HGraph}$ can be realized as a comma category: take as $L$ the functor $\abs{-}:\tree\rightarrow \catname{Set}$ of \cref{forget}, while as $R$ we take $\catname{Set}\to \catname{Set}$ which sends $V$ to $V\times V$ and $f$ to $f\times f$. Applying \cref{comma} we get the following result.
\begin{thm}\label{hiergraph}
	$\catname{HGraph}$ is an adhesive category.
\end{thm}
Let $\mathcal{G}$ be a hierarchical graph, we can model an \emph{interface} as a function between a set $X$ and the set of nodes $V$. Now,  $\abs{-}:\tree\rightarrow \catname{Set}$ preserves the initial objects, thus, by \cref{prop:left}, the forgetful functor $\catname{HGraph}\to \catname{Set}$, which only remembers the set of nodes, has a left adjoint $\Delta$, thus an interface is just a morphism $\Delta(X)\to \mathcal{G}$. This suggests the definition of the following category.
\begin{defi}
	The category $\catname{HIGraph}$ of \emph{hierarchical graphs with interface} is the category $\comma{\Delta}{\id{\catname{HGraph}}}$.
\end{defi}

We can give a more explicit description of $\catname{HIGraph}$. Objects are triples $(\mathcal{G}, X, f)$ made by a hierarchical graph $\mathcal{G}$, a set $X$ and a function $f:X\rightarrow E_\mathcal{G}$. A morphism $(\mathcal{G}, X, f)\rightarrow (\mathcal{G}, Y, g)$ is a triple $(h,k,l)$ with $f:(E, \leq)\rightarrow (F, \leq)$ in $\tree$, $g:V_\mathcal{G}\rightarrow V_\mathcal{H}$  and $l:X\rightarrow Y$ in $\catname{Set}$ such that the following squares commute
\begin{center}
	\begin{tikzpicture}
	\node(A) at(0,0){$E_\mathcal{G}$};
	\node(B) at (1.5,0){$V_\mathcal{G}$};
	\node(C) at(1.5,-1.5){$V_\mathcal{H}$};
	\node(D) at (0,-1.5){$E_\mathcal{H}$};
	\draw[->](A)--(B)node[pos=0.5, above]{$s_\mathcal{G}$};
	\draw[->](D)--(C)node[pos=0.5, below]{$s_\mathcal{H}$};
	\draw[->](A)--(D)node[pos=0.5, left]{$\abs{h}$};
	\draw[->](B)--(C)node[pos=0.5, right]{$k$};
	\node(A) at(3,0){$E_\mathcal{G}$};
	\node(B) at (4.5,0){$V_\mathcal{G}$};
	\node(C) at(4.5,-1.5){$V_\mathcal{H}$};
	\node(D) at (3,-1.5){$E_\mathcal{H}$};
	\draw[->](A)--(B)node[pos=0.5, above]{$t_\mathcal{G}$};
	\draw[->](D)--(C)node[pos=0.5, below]{$t_\mathcal{H}$};
	\draw[->](A)--(D)node[pos=0.5, left]{$\abs{h}$};
	\draw[->](B)--(C)node[pos=0.5, right]{$k$};
	\node(A) at(6,0){$X$};
	\node(B) at (7.5,0){$V_\mathcal{G}$};
	\node(C) at(7.5,-1.5){$V_\mathcal{H}$};
	\node(D) at (6,-1.5){$Y$};
	\draw[->](A)--(B)node[pos=0.5, above]{$f$};
	\draw[->](D)--(C)node[pos=0.5, below]{$g$};
	\draw[->](A)--(D)node[pos=0.5, left]{$l$};
	\draw[->](B)--(C)node[pos=0.5, right]{$k$};
	\end{tikzpicture}
\end{center}

Whatever description we choose, the following result now follows from \cref{comma}.
\begin{thm}\label{interface}
	$\catname{HIGraph}$ is an adhesive category.
\end{thm}
\section{Application to some categories of hypergraphs}\label{sec:hyper}
In this section we will move from the world of graphs to the one of \emph{hypergraphs} allowing an edge to join two arbitrary subsets of nodes. Even in this case,  leveraging the modularity provided by \cref{func}, it is possible to combine sufficiently adhesive categories of preorders or graphs (modelling the hierarchy between the edges) while retaining suitable adhesivity properties.  It is worth noticing that, beside hypergraphs or interfaces, this methodology can be extended easily to other settings such as Petri nets~(see \cite{ehrig1991parallelism}).

\subsection{Hypergraphs}\label{sub:hyper}
We will start this section with the definition of (directed) hypergraph and we will see how label them with an algebraic signature. We will denote by $(-)^\star$ the monad $\catname{Set}\to \catname{Set}$ associated to the algebraic theory of monoids (i.e. the \emph{Kleene star}), moreover, given a set V, $e_V$ will be the empty word in $V^{\star}$

\begin{defi}A \emph{hypergraph} is a 4-uple $\mathcal{G}:(E_\mathcal{G}, V_\mathcal{G}, s_\mathcal{G}, t_\mathcal{G})$ given by two sets $E_\mathcal{G}$ and $V_\mathcal{G}$, whose elements are called respectively \emph{hyperedges} and \emph{nodes}, pluse two \emph{source} and \emph{target}  functions $s_\mathcal{G}, t_\mathcal{G}:E_\mathcal{G}\rightrightarrows V_\mathcal{G}^\star$. A \emph{hypergraph morphism} $(E_\mathcal{G}, V_\mathcal{G}, s_\mathcal{G}, t_\mathcal{G})\to (E_\mathcal{H}, V_\mathcal{H}, s_\mathcal{H}, t_\mathcal{H})$ is a pair $(h,k)$ of functions $h:E_\mathcal{G}\to E_\mathcal{H}$, $k:V_\mathcal{G}\to V_\mathcal{H}$ such that the following diagram commute.
	 \begin{center}
	 	\begin{tikzpicture}
	 	\node(A) at(2,0){$E_\mathcal{G}$};
	 	\node(B) at (3.5,0){$V_\mathcal{G}^{\star}$};
	 	\node(C) at(3.5,-1.5){$\mathcal{V_\mathcal{H}}^{\star}$};
	 	\node(D) at (2,-1.5){$E_\mathcal{H}$};
	 	\draw[->](A)--(B)node[pos=0.5, above]{$s_\mathcal{G}$};
	 	\draw[->](D)--(C)node[pos=0.5, below]{$s_\mathcal{H}$};
	 	\draw[->](A)--(D)node[pos=0.5, left]{$h$};
	 	\draw[->](B)--(C)node[pos=0.5, right]{$k^{\star}$};
	 	\node(A) at(5,0){$E_\mathcal{G}$};
	 	\node(B) at (6.5,0){$V_\mathcal{G}^{\star}$};
	 	\node(C) at(6.5,-1.5){$V_\mathcal{H}^{\star}$};
	 	\node(D) at (5,-1.5){$E_\mathcal{H}$};
	 	\draw[->](A)--(B)node[pos=0.5, above]{$t_\mathcal{G}$};
	 	\draw[->](D)--(C)node[pos=0.5, below]{$t_\mathcal{H}$};
	 	\draw[->](A)--(D)node[pos=0.5, left]{$h$};
	 	\draw[->](B)--(C)node[pos=0.5, right]{$k^\star$};
	 	\end{tikzpicture}
	 \end{center}	 
 We define $\hyp$ to be the resulting category.
\end{defi}

\begin{notaz}Given a set $X$, $\lgt_X:X^{\star}\to \mathbb{N}$ is the function which sends a word to its length. Notice that for every function $f:X\to Y$, the following diagram commutes
		\begin{center}
		\begin{tikzpicture}
		\node(A) at(0,0){$X^\star$};
		\node(B) at (3,0){$Y^\star$};
		\node(C) at(1.5,-1.5){$\mathbb{N}$};
		\draw[->](A)--(B)node[pos=0.5, above]{$f^\star$};
		\draw[->](A)--(C)node[pos=0.5, left, xshift=-0.1cm]{$\lgt_X$};
		\draw[->](B)--(C)node[pos=0.55, right, xshift=0.1cm]{$\lgt_Y$};	
		\end{tikzpicture}
	\end{center}
$e_X$ will denote the empty word in $X^\star$, moreover given $x\in X^\star\smallsetminus\{e_X\}$ and $1\leq n\leq \lgt_X(x)$, $x_n$ is the $n^\mathrm{th}$ letter of $x$.
\end{notaz}

It's easy to see that this definition is exactly the definition of the comma category $\comma{\id{\catname{Set}}}{R}$ where $R:\catname{Set}\to \catname{Set}$ is the functor
\begin{align*}
\functor[l]{X}{f}{Y}
\functormapsto
\rfunctor{X^\star \times X^\star}{f^\star \times f^\star }{Y^\star \times Y^{\star}}
\end{align*}  

We can also notice that the monoid monad $(-)^{\star}:\catname{Set}\to \catname{Set}$ is \emph{cartesian}, i.e. preserves all connected limits. This in turn rests upon the fact that the theory of monoids is a \emph{strongly regular theory} (see \cite[Sec. 3]{carboni1995connected}  and \cite[Ch.4]{leinster2004higher} for details). In particular it preserves pullbacks, thus we can apply \cref{comma,cor:mono}.

\begin{prop}\label{prop:hypadh}
		$\hyp$ is an adhesive category.
\end{prop}

\begin{rem}\label{rem:mono}
Preservation of connected limits implies that $(-)^\star$ sends monos to monos.
\end{rem}

 \cref{prop:left} allows us to deduce immediately the following.

\begin{prop}\label{cor:left}
	The forgetful functor $U_{\catname{Hyp}}:\hyp \to \catname{Set}$ which sends an hypergraph $\mathcal{G}$ to its set of nodes has a left adjoint $\Delta_{\hyp}$.
\end{prop}

\begin{rem}Since the initial object of $\catname{Set}$ is the empty set,  $\Delta_{\hyp}(X)$ is the hypergraph which has $X$ as set of nodes and $\emptyset$ as set of hyperedges.
\end{rem}

Hypergraphs, as normal graphs, can be represented graphically. We will use dots to denote nodes and squares to denote hyperedges, the name of a node or of an hyperedge will be put near the corresponding dot or square. Sources and targets are represented by lines between dots and squares: the lines from the sources of an hyperedge will enter its square from the left, while the lines to the targets will exit it from the right, we will adopt the convention for which sources and targets are ordered from the top to the bottom. We can now illustrate this giving some example.
\begin{exa}Take $V$ to be be $\{v_1, v_2, v_3, v_4, v_5\}$ and $E$ to be $\{h_1, h_2, h_3\}$. Sources and targets are given by:
\[	s(h_1)=v_1v_2\quad
s(h_2)=v_3v_4\quad
s(h_3)=v_5 \qquad 
t(h_1)=v_3v_4\quad 
t(h_2)=v_5\quad
t(h_3)=e_V
\]	
	We can draw the resulting $\mathcal{G}$ as follows:
	
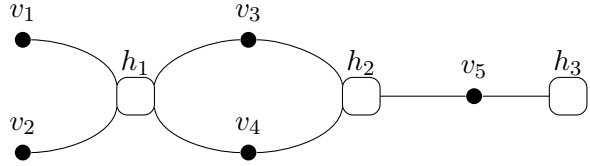
\begin{center}\begin{tikzpicture}
		\node[circle,fill=black,inner sep=0pt,minimum size=6pt,label=above:{$v_1$}] (A) at (0,0) {};
		\node[circle,fill=black,inner sep=0pt,minimum size=6pt,label=above:{$v_2$}] (B) at (0,-1.5) {};
		\node[circle,fill=black,inner sep=0pt,minimum size=6pt,label=above:{$v_3$}] (C) at (3,0) {};
		\node[circle,fill=black,inner sep=0pt,minimum size=6pt,label=above:{$v_4$}] (D) at (3,-1.5) {};
		\node[circle,fill=black,inner sep=0pt,minimum size=6pt,label=above:{$v_5$}] (E) at (6,-0.75) {};
		\draw[rounded corners] (1.25, -1) rectangle (1.75, -0.5) {};
		\draw[-](4.75,-0.75)--(E);
		\draw[-](7,-0.75)--(E);
		\draw[rounded corners] (4.25, -1) rectangle (4.75, -0.5) {};
		\node at (4.5, -0.3){$h_2$};
		\node at (1.5, -0.3){$h_1$};
		\node at (7.25, -0.3){$h_3$};
		\draw[rounded corners] (7, -1) rectangle (7.5, -0.5) {};
		\draw(A)..controls(0.5,0)and(1.2,-0.2)..(1.25,-0.6);
		\draw(B)..controls(0.5,-1.5)and(1.2,-1.3)..(1.25,-0.9);
		
		\draw(C)..controls(3.5,0)and(4.2,-0.25)..(4.25,-0.6);
		\draw(D)..controls(3.5,-1.5)and(4.2,-1.3)..(4.25,-0.9);
		
		\draw(D)..controls(2.5,-1.5)and(1.8,-1.3)..(1.75,-0.9);
		\draw[-](C)..controls(2.5,0)and(1.8,-0.25)..(1.75,-0.6);
		\end{tikzpicture}
	\end{center}
\end{exa}
\begin{exa}\label{exa_2} Let $V$ be as in the previous example and $E=\{h_1, h_2, h_3\}$.	Then we define
\[
s(h_1)=e_V \quad s(h_2)=v_1v_2\quad s(h_3)=v_3v_4\qquad t(h_1)=v_1 \quad t(h_2)=v_3\quad t(h_3)=v_5
\]	
	Now we can depict $\mathcal{G}$ as
	\begin{center}\begin{tikzpicture}
		\node[circle,fill=black,inner sep=0pt,minimum size=6pt,label=above:{$v_1$}] (A) at (0,0) {};
		\node[circle,fill=black,inner sep=0pt,minimum size=6pt,label=above:{$v_2$}] (B) at (0,-1.5) {};
		\node[circle,fill=black,inner sep=0pt,minimum size=6pt,label=above:{$v_3$}] (C) at (3,-0.75) {};
		\node[circle,fill=black,inner sep=0pt,minimum size=6pt,label=above:{$v_4$}] (D) at (3,-2.25) {};
		\node[circle,fill=black,inner sep=0pt,minimum size=6pt,label=right:{$v_5$}] (E) at (6,-1.5) {};
		\draw[-](1.75,-0.75)--(C);
		\draw[rounded corners] (1.25, -1) rectangle (1.75, -0.5) {};
		\draw[-](4.75,-1.5)--(E);
		\draw[-](-1.5,0)--(A);
		\draw[rounded corners] (4.25, -1.75) rectangle (4.75, -1.25) {};
		\node at (4.5, -1.05){$h_3$};
		\node at (1.5, -0.3){$h_2$};
		\node at (-1.75, 0.45){$h_1$};
		\draw[rounded corners] (-2, -0.25) rectangle (-1.5, 0.25) {};
		\draw(A)..controls(0.5,0)and(1.2,-0.2)..(1.25,-0.6);
		\draw(B)..controls(0.5,-1.5)and(1.2,-1.3)..(1.25,-0.9);
		
		\draw(C)..controls(3.5,-0.75)and(4.2,-0.95)..(4.25,-1.35);
		\draw(D)..controls(3.5,-2.25)and(4.2,-2.05)..(4.25,-1.65);
		\end{tikzpicture}
	\end{center}
\end{exa}

\begin{exa}\label{exa_3} Let $\Sigma=(O, \ari)$ be an algebraic signature ($O$ is a set and $\ari:O\rightarrow \mathbb{N}$ a function called \emph{arity function}), we can construct the hypergraph $\mathcal{G}^\Sigma$ taking $V$ and $E$ to be respectively the singleton ${v}$ and the set $O$. We put
	\begin{align*}
s_{\mathcal{G}^\Sigma}:O\to \{v\}^\star \quad o\mapsto v^{\ari(o)}
\qquad 
t_{\mathcal{G}^\Sigma}:O\to \{v\}^\star \quad o\mapsto v
\end{align*}
	
	For instance let $\Sigma$ be the signature of groups $(\{m, i, e \}, \ari)$ with
	\[\ari(m)=2 \quad \ari(i)=1 \quad \ari(e)=0\]	
	Then $\mathcal{G}^\Sigma$ is depicted as:
	\begin{center}
		\begin{tikzpicture}
		\node[circle,fill=black,inner sep=0pt,minimum size=6pt,label=above:{$v$}] (V) at (0,0) {};
		\node(E)at(-2, 1.4){$e$};
		\node(M)at(0, 2.15){$m$};
		\node(I)at(2, 1.5){$i$};
		\draw[-](V)..controls(-1.2,0.1)..(-1.75,1);
		\draw[-](V)..controls(-0.5,0.5)and(-0.8,1)..(-0.25,1.6);
		\draw[-](V)..controls(-1,0.6)and(-1,1.1)..(-0.25,1.9);
		\draw[-](V)..controls(0.5,0.5)and(0.8,0.8)..(0.25,1.75);
		\draw[-](V)..controls(1.2,0.1)..(1.75,1);
		\draw[-](2.25,1)..controls(2.8,1)and(2.5,0)..(V);
		\draw[rounded corners] (-2.25, 0.75) rectangle (-1.75, 1.25) {};
		\draw[rounded corners] (-0.25, 1.5) rectangle (0.25, 2) {};
		\draw[rounded corners] (2.25, 0.75) rectangle (1.75, 1.25) {};
		\end{tikzpicture}
	\end{center}
\end{exa}

This last example is useful in order to label hyperedges with operations.
\begin{defi}Let $\Sigma=(O, \ari)$ be an algebraic signature, the category $\hyps$ of \emph{labeled hypergraphs} is the slice category $\hyp/G^\Sigma$.
\end{defi}
\cref{cor:mono} and \cref{cor:slice} give us immediately an adhesivity result for $\hyp_{\Sigma}$ and a characterization of monomorphisms in it.
\begin{prop}\label{prop:mono}
	For every algebraic signature $\Sigma$, $\hyps$ is an adhesive category. Moreover a morphism $(h,k)$ between two object of $\hyp_{\Sigma}$ is a mono if and only if $h$ and $k$ are injective functions.
\end{prop}

$\hyp_{\Sigma}$ has a forgetful functor $U_{\Sigma}:\hyp_{\Sigma}\to \catname{Set}$ which sends $h:\mathcal{H}\to \mathcal{G}^{\Sigma}$ to $U_{\hyp}(\mathcal{H}$). Now, $U_{\hyp}(\mathcal{G}^{\Sigma})=\{v\}$ thus, for every set $X$, there is only one arrow $X\to U_{\hyp}(\mathcal{G}^{\Sigma})$. Define $\Delta_{\Sigma}(X):\Delta_{\hyp}(X)\to \mathcal{G}^{\Sigma}$ to be the transpose of this arrow.

\begin{prop} $U_\Sigma$
	has a left adjoint $\Delta_\Sigma$.
\end{prop}
\begin{proof}Let $h:\mathcal{H}\to \mathcal{G}^{\Sigma}$ be an object of $\hyp_{\Sigma}$, and suppose that there exists $f:X\to U_{\Sigma}(\mathcal{H})$. Since $U_{\Sigma}(\mathcal{H})=U_{\hyp}(\mathcal{H})$ and $\id{\catname{Set}}$ is the unit of $\Delta\dashv U_{\hyp}$, there exists a unique morphism of $\hyp$ $(k,f):\Delta_{\hyp}(X)\to \mathcal{H}$. Since the set of hyperedges of $\Delta_{\hyp}(X)$ is empty, $k$ must be the empty function and the commutativity of each of the two triangles below is equivalent to that of the other
	\begin{center}
		\begin{tikzpicture}
		\node(A) at(0,0){$\Delta_{\hyp}(X)$};
		\node(B) at (3,0){$\mathcal{H}$};
		\node(C) at(1.5,-1.5){$\mathcal{G}^{\Sigma}$};
		\draw[->](A)--(B)node[pos=0.5, above]{$(k,f)$};
		\draw[->](A)--(C)node[pos=0.5, left, xshift=-0.1cm]{$\Delta_\Sigma(X)$};
		\draw[->](B)--(C)node[pos=0.55, right, xshift=0.1cm]{$(h, !_{V_\mathcal{H}})$};

		\node(A) at(6,0){$U_{\hyp}\qty(\Delta_{\hyp}(X))$};
		\node(B) at (9,0){$U_{\hyp}\qty(\mathcal{H})$};
		\node(C) at(7.5,-1.5){$U_{\hyp}\qty(\mathcal{G}^{\Sigma})$};
		\draw[->](A)--(B)node[pos=0.5, above]{$f$};
		\draw[->](A)--(C)node[pos=0.5, left, xshift=-0.1cm]{$U_{\hyp}(\Delta_\Sigma(X))$};
		\draw[->](B)--(C)node[pos=0.55, right, xshift=0.1cm]{$h$};	
		\end{tikzpicture}
	\end{center}
	But the triangle on the right commutes because $U_{\hyp}(\mathcal{G}^{\Sigma})$ is terminal.
\end{proof}

A more concrete definition of a labeled hypergraphs can be given. Let $\mathcal{H}=(E, V, s, t)$ be an hypergraph, since  $U_{\hyp}(\mathcal{G}^{\Sigma})$ is the singleton an arrow $\mathcal{H}\rightarrow \mathcal{G}^{\Sigma}$, is determined by a function $f:E\to O$  such that $\ari(h(e))$ is equal to the length of $s(e)$. 

\begin{rem}\label{rem:label}
	If  $\mathcal{H}$ has an hyperedge $h$ such that $t_{\mathcal{H}}(h)$ has a length different from $1$, then there is no morphism $\mathcal{H}\to \mathcal{G}^{\Sigma}$. Indedd, if such a morphism $(f,!_{V_\mathcal{H}}):\mathcal{H}\to \mathcal{G}^\Sigma$ exists, then, for every $h\in E_{\mathcal{H}}$ we have
	\[f^{\star}(t_{\mathcal{H}}(h))=t_{\mathcal{G}^{\Sigma}}(f(h))=v\] 
\end{rem}

We will extend our graphical notation of hypergraph to labeled ones putting the label of an hyperedge $h$ inside its corresponding square.
\begin{exa}\label{lab_1}
	The simplest example is given by the identity $\id{\mathcal{G}^\Sigma}:\mathcal{G}^\Sigma\rightarrow \mathcal{G}^{\Sigma}$. If $\Sigma$ is the signature of groups we get 
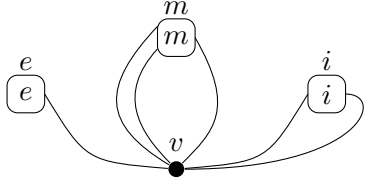
\begin{center}
		\begin{tikzpicture}
		\node[circle,fill=black,inner sep=0pt,minimum size=6pt,label=above:{$v$}] (V) at (0,0) {};
		\node at(-2,1){$e$};
		\node at(0,1.75){$m$};	
		\node at(2,1){$i$};	
		\node(E)at(-2, 1.4){$e$};
		\node(M)at(0, 2.15){$m$};
		\node(I)at(2, 1.45){$i$};
		\draw[-](V)..controls(-1.2,0.1)..(-1.75,1);
		\draw[-](V)..controls(-0.5,0.5)and(-0.8,1)..(-0.25,1.6);
		\draw[-](V)..controls(-1,0.6)and(-1,1.1)..(-0.25,1.9);
		\draw[-](V)..controls(0.5,0.5)and(0.8,0.8)..(0.25,1.75);
		\draw[-](V)..controls(1.2,0.1)..(1.75,1);
		\draw[-](2.25,1)..controls(2.8,1)and(2.5,0)..(V);
		\draw[rounded corners] (-2.25, 0.75) rectangle (-1.75, 1.25) {};
		\draw[rounded corners] (-0.25, 1.5) rectangle (0.25, 2) {};
		\draw[rounded corners] (2.25, 0.75) rectangle (1.75, 1.25) {};
		\end{tikzpicture}
	\end{center}
\end{exa}

\begin{exa}\label{lab_2}
	Take again $\Sigma$ the signature of groups, then the hypergraph $\mathcal{G}$ of \cref{exa_2} can be labeled defining
	\begin{align*}
	f(h_1)=(e) \quad f(h_2)=f(h_3)=m
	\end{align*}
	In this case we get the following:
	
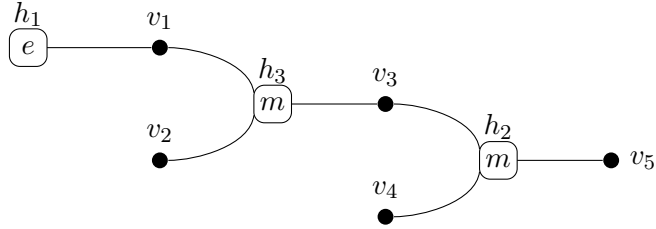
\begin{center}\begin{tikzpicture}
		\node[circle,fill=black,inner sep=0pt,minimum size=6pt,label=above:{$v_1$}] (A) at (0,0) {};
		\node[circle,fill=black,inner sep=0pt,minimum size=6pt,label=above:{$v_2$}] (B) at (0,-1.5) {};
		\node[circle,fill=black,inner sep=0pt,minimum size=6pt,label=above:{$v_3$}] (C) at (3,-0.75) {};
		\node[circle,fill=black,inner sep=0pt,minimum size=6pt,label=above:{$v_4$}] (D) at (3,-2.25) {};
		\node[circle,fill=black,inner sep=0pt,minimum size=6pt,label=right:{$v_5$}] (E) at (6,-1.5) {};
		\node at (-1.75,0) {$e$};
		\node at (1.5,-0.75) {$m$};
		\node at (4.5,-1.5) {$m$};
		\draw[-](1.75,-0.75)--(C);
		\draw[rounded corners] (1.25, -1) rectangle (1.75, -0.5) {};
		\draw[-](4.75,-1.5)--(E);
		\draw[-](-1.5,0)--(A);
		\draw[rounded corners] (4.25, -1.75) rectangle (4.75, -1.25) {};
		\node at (4.5, -1.05){$h_2$};
		\node at (1.5, -0.3){$h_3$};
		\node at (-1.75, 0.45){$h_1$};
		\draw[rounded corners] (-2, -0.25) rectangle (-1.5, 0.25) {};
		\draw(A)..controls(0.5,0)and(1.2,-0.2)..(1.25,-0.6);
		\draw(B)..controls(0.5,-1.5)and(1.2,-1.3)..(1.25,-0.9);
		
		\draw(C)..controls(3.5,-0.75)and(4.2,-0.95)..(4.25,-1.35);
		\draw(D)..controls(3.5,-2.25)and(4.2,-2.05)..(4.25,-1.65);
		\end{tikzpicture}
	\end{center}
\end{exa}

\begin{rem}There is a \emph{colored} (or \emph{typed}) version of these last constructions. Start with a \emph{colored} algebraic signature: this is a triple $(C, O, \ari)$ where $C$ is the set of \emph{colors}, $O$ is the set of \emph{operations} and $\ari:O\rightarrow C^{\bullet}\times C^{\bullet}$ assigns to every operations $f$ an arity and a coarity given by strings of colors.  We can still construct an hypergraph $\mathcal{G}^{\Sigma}$ with $C$ as set of nodes using the operations as hyperedges.  In this context an object in the slice $\hyp/\mathcal{G}^{\Sigma}$ is an hypergraph in which both the hyperedges and the nodes are labeled, the formers with an elemento of $O$ and the latters with an element of $C$ \cite{bonchi2022string}. 
\end{rem}
\paragraph{$\hyp$ as a topos of presheaves}
By \cref{lim} we already know that $\hyp$ has all connected limits, and by \cref{prop:hypadh} we know that it is adhesive. Actually more can be proved about it: we can realize $\hyp$ as a presheaf topos \cite{bonchi2022string}.

\begin{defi}Let $\catname{I}$ be the category in which:
	\begin{itemize}
		\item the set of objects is given by $\qty(\mathbb{N}\times \mathbb{N}) \cup \{\bullet\}$
		\item arrows are given by the identities $\id{k,l}$ and $\id{\bullet}$ and exactly $k+l$ arrows $f_i:(k,l)\rightarrow \bullet$;
		\item composition is defined simply putting, for every $f_i:(k,l)\rightarrow \bullet$:
		\begin{equation*}
		f_i\circ \id{k,l}=f_i = \id{\bullet}\circ f_i 
		\end{equation*}
	\end{itemize}
	
\end{defi}

Now, given $F:\catname{I}\to \catname{Set}$ we can define
\[E_F:=\bigsqcup_{k,l\in \mathbb{N}}F(k,l)\]
Then we have
\begin{align*}
	s_{k,l}&:F(k,l)\to F(\bullet)^{\star} \qquad x \mapsto \prod_{i=1}^{k}F(f_i)(x)\\
	t_{k,l}&:F(k,l)\to F(\bullet)^{\star} \qquad x \mapsto \prod_{i=k+1}^{k+l}F(f_i)(x)
\end{align*}
which induce
 $s_F, t_F:E_F\rightrightarrows \mathcal{F}(\bullet)^{\star}$. Let $\mathcal{G}_F$ be the resulting hypergraph. Now, every $\eta:F\rightarrow H$ in $\catname{Set}^{\catname{I}}$ has components $\eta_{k,l}:F(k,l)\to H(k,l)$, $\eta_{\bullet}:F(\bullet)\to H(\bullet)$, thus it induces a function $\hat{\eta}:E_F\rightarrow E_H$ such that the following squares commute
\begin{center}
	\begin{tikzpicture}
	\node(A) at(2,0){$E_F$};
	\node(B) at (3.5,0){$F(\bullet)^{\star}$};
	\node(C) at(3.5,-1.5){$H(\bullet)^{\star}$};
	\node(D) at (2,-1.5){$E_{H}$};
	\draw[->](A)--(B)node[pos=0.5, above]{$s_F$};
	\draw[->](D)--(C)node[pos=0.5, below]{$s_H$};
	\draw[->](A)--(D)node[pos=0.5, left]{$\hat{\eta}$};
	\draw[->](B)--(C)node[pos=0.5, right]{$\eta_{\bullet}^{\star}$};
	\node(A) at(5,0){$E_F$};
	\node(B) at (6.5,0){$F(\bullet)^{\star}$};
	\node(C) at(6.5,-1.5){$H(\bullet)^{\star}$};
	\node(D) at (5,-1.5){$E_F$};
	\draw[->](A)--(B)node[pos=0.5, above]{$t_F$};
	\draw[->](D)--(C)node[pos=0.5, below]{$t_H$};
	\draw[->](A)--(D)node[pos=0.5, left]{$\hat{\eta}$};
	\draw[->](B)--(C)node[pos=0.5, right]{$\eta_{\bullet}^{\star}$};
	\end{tikzpicture}
\end{center}
this is equivalent to say that $\eta$ induces a morphism $(\hat{\eta}, \eta_{\bullet}):\mathcal{G}_F\to \mathcal{G}_H$. It is now clear that sending $F$ to $\mathcal{G}_F$ and $\eta$ to $(\hat{\eta}, \eta_{\bullet})$ defines a faithful functor $\mathcal{G}_{-}:\catname{Set}^{\catname{I}}\to \hyp$.

\begin{prop}
	$\hyp$ is equivalent to the category $\catname{Set}^{\catname{I}}$.
\end{prop}
\begin{proof}
	Let $X$ be a set, for every $n\in \mathbb{N}$ define
	\[X_{n}:=\qty{v\in X^{\star} \mid \lgt_X(v)=n}\]

	Given $F:\catname{I}\to \catname{Set}$ we have that
	\[i^F_{k,l}\qty(F(k,l))=s^{-1}_F\qty(F(\bullet)_{k})\cap t^{-1}_F\qty(F(\bullet)_{l}) \]
	where $i^{F}_{k,l}$ is the inclusion $F(k,l)\to E_F$.
	
	We are now ready to that $\mathcal{G}_{-}$ is full and essentially surjective.
	\begin{itemize}
		\item For fullness, let $(f,g):\mathcal{G}_F\to \mathcal{G}_{H}$ be a morphism of hypergraphs and define $h_{k,l}$ to be $h\circ i^F_{k,l}$, the composition of $h$ with Now, if $x\in F(k,l)$ then 
		\begin{gather*}s_H\qty(h_{k,l}\qty(x))=s_H\qty(h\qty(i^F_{k,l}\qty(x)))=g^{\star}\qty(s_{F}\qty(i^F_{k,l}\qty(x)))\\t_H\qty(h_{k,l}\qty(x))=s_t\qty(h\qty(i^F_{k,l}\qty(x)))=g^{\star}\qty(t_{F}\qty(i^F_{k,l}\qty(x)))
		\end{gather*}
	Now, $\lgt\qty(g^{\star}(v))=\lgt\qty(v)$ for every $v\in F(\bullet)^{\star}$, thus the previous computations shows that there exists $\eta_{k,l}$ such that the square
	\begin{center}
		\begin{tikzpicture}
		\node(A) at(2,0){$F(k,l)$};
		\node(B) at (4,0){$E_F$};
		\node(C) at(4,-1.5){$E_H$};
		\node(D) at (2,-1.5){$H(k,l)$};
		\draw[->](A)--(B)node[pos=0.5, above]{$i^{F}_{k,l}$};
		\draw[->](D)--(C)node[pos=0.5, below]{$i^{H}_{k,l}$};
		\draw[->](A)--(D)node[pos=0.5, left]{$\eta_{k,l}$};
		\draw[->](B)--(C)node[pos=0.5, right]{$h_{k,l}$};
		\end{tikzpicture}
	\end{center}
	commutes.  Now, defining $\eta_{\bullet}$ as $g$, the collection $\qty{\eta_{k,l}}_{k,l\in \mathbb{N}}$, defines a natural transformation $\eta:F\to H$. Indeed, if $f_i:(k,l)\to \bullet$ we have:
	\begin{center}
	\begin{tikzpicture}
	\node(A) at(2,0){$F(k,l)$};
	\node(B) at (4,0){$E_F$};
	\node(C) at(4,-1.5){$E_H$};
	\node(D) at (2,-1.5){$H(k,l)$};
	\node(E) at (5.5, 0){$F(\bullet)^\star$};
	\node(F) at (5.5, -1.5){$F(\bullet)^\star$};
	\draw[->](A)--(B)node[pos=0.5, above]{$i^{F}_{k,l}$};
	\draw[->](D)--(C)node[pos=0.5, below]{$i^{H}_{k,l}$};
	\draw[->](A)--(D)node[pos=0.5, left]{$\eta_{k,l}$};
	\draw[->](B)--(C)node[pos=0.5, left]{$h_{k,l}$};
	\draw[->](B)--(E)node[pos=0.5, above]{$s_F$};
\draw[->](C)--(F)node[pos=0.5, below]{$s_H$};
\draw[->](E)--(F)node[pos=0.5, right]{$g$};

\node(A) at(7,0){$F(k,l)$};
\node(B) at (9,0){$E_F$};
\node(C) at(9,-1.5){$E_H$};
\node(D) at (7,-1.5){$H(k,l)$};
\node(E) at (10.5, 0){$F(\bullet)^\star$};
\node(F) at (10.5, -1.5){$F(\bullet)^\star$};
\draw[->](A)--(B)node[pos=0.5, above]{$i^{F}_{k,l}$};
\draw[->](D)--(C)node[pos=0.5, below]{$i^{H}_{k,l}$};
\draw[->](A)--(D)node[pos=0.5, left]{$\eta_{k,l}$};
\draw[->](B)--(C)node[pos=0.5, left]{$h_{k,l}$};
\draw[->](B)--(E)node[pos=0.5, above]{$t_F$};
\draw[->](C)--(F)node[pos=0.5, below]{$t_H$};
\draw[->](E)--(F)node[pos=0.5, right]{$g$};
	\end{tikzpicture}
\end{center}
	The diagram on the right implies naturality where $i\leq k$, while the one on the left takes care of the other case. Finally, by contruction it is clear that $(\hat{\eta}, \eta_{\bullet})=(f,g)$. 
		\item Given an hypergraph $\mathcal{G}=(E, V, s, t)$ we can define 
		\[F_{\mathcal{G}}(k,l):=s^{-1}(V_k)\cap t^{-1}(V_l) \qquad F_{\mathcal{G}}(\bullet):=V\]
		Given $f_i:(k,l)\to \bullet$ we put
	\[F_{\mathcal{G}}(f_i):F_{\mathcal{G}}(k,l)\to F_{\mathcal{G}}(\bullet) \qquad x\mapsto \begin{cases}
	s(x)_{i} & i\leq k\\
	t(x)_{i-l} &k <i
	\end{cases}  \]
	
Now, $F_{\mathcal{G}}$ is a functor $\catname{I}\to \catname{Set}$ and for every $h\in E$ there exists a unique pair $(k,l)$ such that $h\in F_{\mathcal{G}}(k,l) $, thus
\[\bigsqcup_{k,l\in \mathbb{N}}F_{\mathcal{G}}(k,l)\simeq E\]
Moreover, by construction $s_{F_{\mathcal{G}}}=s$ and $t_{F_{\mathcal{G}}}=t$, from which the thesis follows. \qedhere
	\end{itemize}
\end{proof}
As a corollary we get immediately the following.
\begin{cor}
	$\hyp$ is a complete category.
\end{cor}

\subsection{Hierarchical hypergraphs}\label{subsec:hhgraph}  We can leverage on the modularity of \cref{func} and \cref{comma} to give hypergraphical variants for \cref{hiergraph} and \cref{interface}. This is done replacing the set $E_\mathcal{G}$ of hyperedges with a tree order $(E_\mathcal{G}, \leq)$ and $\id{\catname{Set}}$ with the forgetful functor $\abs{-}:\tree\to \catname{Set}$.

\begin{defi}A \emph{hierarchical hypergraph} $\mathcal{G}$ is a triple $((E_\mathcal{G}, \leq), V_\mathcal{G}, e_\mathcal{G})$ where $(E_\mathcal{G}, \leq)$ is a tree order, $V_\mathcal{G}$ a set and $e_\mathcal{G}:E_\mathcal{G}\rightarrow V_\mathcal{G}^{\star}$ a function. A \emph{morphism} $\mathcal{G}\rightarrow \mathcal{H}$  is a pair $(f,g)$ with $f:(E_\mathcal{G}, \leq)\rightarrow (E_\mathcal{H}, \leq)$ in $\tree$, $g:V\rightarrow W$ in $\catname{Set}$ such that the following square commutes
	\begin{center}
		\begin{tikzpicture}
		\node(A) at(2,0){$E_\mathcal{G}$};
		\node(B) at (3.5,0){$V_\mathcal{G}^\star$};
		\node(C) at(3.5,-1.5){$V_\mathcal{H}^\star$};
		\node(D) at (2,-1.5){$E_\mathcal{H}$};
		\draw[->](A)--(B)node[pos=0.5, above]{$e_\mathcal{G}$};
		\draw[->](D)--(C)node[pos=0.5, below]{$e_\mathcal{H}$};
		\draw[->](A)--(D)node[pos=0.5, left]{$\abs{f}$};
		\draw[->](B)--(C)node[pos=0.5, right]{$g^\star$};
		\end{tikzpicture}
	\end{center}
Taking componentwise composition we get a category $\catname{HHGraph}$.
\end{defi}
It's now easy to see that, with this definition, $\catname{HHGraph}$ is the comma category $\comma{\abs{-}}{(-)^{\star}}$, therefore deducing its adhesivity.
\begin{thm}
	$\catname{HHGRaph}$ is adhesive. Moreover, the functor $\catname{HHGraph}\to \catname{Set}$, which sends a hierarchical hypergraph to its set of nodes, has a left adjoint $\Delta_{\catname{HHGraph}}$.
\end{thm}
\begin{proof}The first half follows from \cref{comma}, the second one from \cref{prop:left}.
\end{proof}

To add interface we proceed exactly as in \cref{subsec:hgraph}, using \cref{cor:left}

\begin{defi}The category $\catname{HHIGraph}$ of \emph{hierarchical hypergraphs with interface} is the comma category $\comma{\Delta_{\catname{HHGraph}}}{\id{\hyp}}$.
\end{defi}
As before we can give a more explicit description of $\catname{HHIGraph}$. An object in it is a triple $(\mathcal{G}, X, f)$ made by a hierarchical hypergraph $\mathcal{G}=((E_\mathcal{G}, \leq), V_\mathcal{G}, e_\mathcal{G})$, a set $X$ and a function $f:X\to V$. A morphism $(\mathcal{G}, X, f)\to (\mathcal{H}, Y, g)$ is a triples $(h,k,l)$  with $h:(E_\mathcal{G}, \leq)\rightarrow (E_\mathcal{H}, \leq)$ in $\tree$, $k:V_\mathcal{G}\rightarrow V_\mathcal{H}$  and $l:X\rightarrow Y$ in $\catname{Set}$ such that the following squares commute
	\begin{center}
		\begin{tikzpicture}
		\node(A) at(2,0){$E_\mathcal{G}$};
		\node(B) at (3.5,0){$V_\mathcal{G}^\star$};
		\node(C) at(3.5,-1.5){$V_\mathcal{H}^\star$};
		\node(D) at (2,-1.5){$E_\mathcal{H}$};
		\draw[->](A)--(B)node[pos=0.5, above]{$e_\mathcal{G}$};
		\draw[->](D)--(C)node[pos=0.5, below]{$e_\mathcal{H}$};
		\draw[->](A)--(D)node[pos=0.5, left]{$\abs{h}$};
		\draw[->](B)--(C)node[pos=0.5, right]{$k^\star$};
		\node(A) at(5,0){$X$};
		\node(B) at (6.5,0){$V_\mathcal{G}$};
		\node(C) at(6.5,-1.5){$V_\mathcal{H}$};
		\node(D) at (5,-1.5){$Y$};
		\draw[->](A)--(B)node[pos=0.5, above]{$f$};
		\draw[->](D)--(C)node[pos=0.5, below]{$g$};
		\draw[->](A)--(D)node[pos=0.5, left]{$l$};
		\draw[->](B)--(C)node[pos=0.5, right]{$k$};
		\end{tikzpicture}
	\end{center}

\begin{rem}This category of hypergraphs whose edges form a tree order, corresponds to Milner's (pure) bigraphs \cite{milner:bigraphs}, with possibly infinite edges\footnote{In bigraph terminology, ``controls'' and ``edges'' correspond to our edges and nodes.}.
\end{rem} 
 
 Given its definition, we deduce at once the following.
\begin{thm}
	$\catname{HHIGraph}$ is adhesive.
\end{thm}

 \subsection{$\dgr$ and $\dg$-hypergraphs}
 We can consider more general relations between edges, besides tree orders.
 An interesting case is when edges form a directed acyclic graph, yielding the category of \emph{$\dg$-hypergraphs}; this corresponds to (possibly infinite) \emph{bigraphs with sharing}, where an edge can have more than one parent, as in \cite{sc:bigraphsharing} (see also \cref{fig:dgdag}, left).
 Even more generally, we can consider any relation between edges, i.e., the edges form a generic directed graph possibly with cycles, yielding the category of \emph{$\dgr$-hypergraphs}. These can be seen as ``recursive bigraphs'', i.e., bigraphs which allow for cyclic dependencies between controls, like in recursive processes; an example is in \cref{fig:dgdag} (right).

\begin{defi} A \emph{$\dgr$-hypergraph} (respectively \emph{$\dg$-hypergraphs}) is a triple
	$(\mathcal{G}, V, e)$ where $\mathcal{G}$ is in $\dgr$ (in $\dg$), $V$ is a set and $e$ a function $E_\mathcal{G}\rightarrow V^{\star}$. A \emph{morphism} of \emph{$\dgr$-hypergraph} (\emph{$\dg$-hypergraphs}) is a pair $((h_1, h_2),k):(\mathcal{G}, V, e)\rightarrow (\mathcal{H}, W, e')$ with $(h_1, h_2):\mathcal{G}\rightarrow \mathcal{H}$ in $\dg$ (in $\dgr$) and $k:V\rightarrow W$ in $\catname{Set}$ such that the following square commute
	\begin{center}
		\begin{tikzpicture}
		\node(A) at(2,0){$E_\mathcal{G}$};
		\node(B) at (3.5,0){$V_\mathcal{G}^\star$};
		\node(C) at(3.5,-1.5){$V_\mathcal{H}^\star$};
		\node(D) at (2,-1.5){$E_\mathcal{H}$};
		\draw[->](A)--(B)node[pos=0.5, above]{$e$};
		\draw[->](D)--(C)node[pos=0.5, below]{$e'$};
		\draw[->](A)--(D)node[pos=0.5, left]{$h_2$};
		\draw[->](B)--(C)node[pos=0.5, right]{$k^\star$};
		\end{tikzpicture}
	\end{center}
Thess data give rise to the categories $\catname{SHGraph}$ and $\catname{DAGHGraph}$ respectively.
\end{defi}

We can realise both $\catname{SHGraph}$ and $\catname{DAGHGraph}$ as comma categories,: take respectively the forgetful functors $\dgr\rightarrow \catname{Set}$ and $\dg \rightarrow \catname{Set}$ on one side and the Kleene star $(-)^\ast$ on the other.

\begin{thm}\label{thm:dg}
	$\catname{SHGraph}$ is adhesive with respect to the classes
	\begin{gather*}
	\{((h_1,h_2), k)\in \arr{SHGraph}\mid (h_1,h_2)\in \reg{\dgr}, k \in \mono{Set}\}\\
	\{((h_1,h_2), k)\in \arr{SHGraph}\mid (h_1,h_2)\in \mono{\dgr}\}
	\end{gather*}
	while $\catname{DAGHGraph}$ is adhesive with respect to the classes	\begin{gather*}
	\{((h_1,h_2), k)\in \arr{DAGHGraph}\mid (h_1,h_2)\in \rtd, k, l \in \mono{Set}\}\\
	\{((h_1,h_2), k)\in \arr{DAGHGraph}\mid (h_1,h_2)\in \mono{\dg}\}
	\end{gather*}
	
	Moreover, the functors $\catname{DHGraph}\to \catname{Set}$ and  $\catname{DAGHGraph}\to \catname{Set}$, which assign to an hypergraph its set of nodes, have left adjoints $\Delta_{\catname{DHGraph}}$  and $\Delta_{\catname{DAGHGraph}}$.
\end{thm}

 As in \cref{subsec:hgraph,subsec:hhgraph}, we can exploit these two last corollaries to add interfaces.
 
 \begin{defi}
 The categories $\catname{SHIGraph}$ and $\catname{DAGIHGraph}$ of, respectively, $\dgr$-hypergraphs and $\dg$-hypergraphs \emph{with interfaces} are defined as $\comma{\Delta_{\catname{SHGraph}}}{\id{\catname{SHGraph}}}$ and $\comma{\Delta_{\catname{DAGHGraph}}}{\id{\catname{DAGHGraph}}}$.
 \end{defi}

If we unravel the definition we get the following description of these two categories.
An object in $\catname{SHIGraph}$ ($\catname{DAGHGraph}$) is a triple $(\mathcal{G}, V, e), X, f)$ where $(\mathcal{G}, V, e)$ is a $\dgr$-hypergraph (a $\dg$-hypergraph) and $f$ is a function $X\rightarrow V$.  An arrow $((\mathcal{G}, V, e), X, f)\to ((\mathcal{H}, w, e'), Y, g)$ is then a triple $((h_1, h_2),k,l)$ made by $(h_1, h_2):\mathcal{G}\rightarrow \mathcal{H}$ in $\dgr$ (in $\dg$),  $k:V\to W$ and $l:X\rightarrow Y$ in $\catname{Set}$ such that the following squares commute
	\begin{center}
		\begin{tikzpicture}
		\node(A) at(2,0){$E_{\mathcal{G}}$};
		\node(B) at (3.5,0){$V^\star$};
		\node(C) at(3.5,-1.5){$W^\star$};
		\node(D) at (2,-1.5){$E_{\mathcal{H}}$};
		\draw[->](A)--(B)node[pos=0.5, above]{$e$};
		\draw[->](D)--(C)node[pos=0.5, below]{$e'$};
		\draw[->](A)--(D)node[pos=0.5, left]{$h_2$};
		\draw[->](B)--(C)node[pos=0.5, right]{$k^\star$};
		\node(A) at(5,0){$X$};
		\node(B) at (6.5,0){$V$};
		\node(C) at(6.5,-1.5){$W$};
		\node(D) at (5,-1.5){$Y$};
		\draw[->](A)--(B)node[pos=0.5, above]{$f$};
		\draw[->](D)--(C)node[pos=0.5, below]{$g$};
		\draw[->](A)--(D)node[pos=0.5, left]{$l$};
		\draw[->](B)--(C)node[pos=0.5, right]{$k$};
		\end{tikzpicture}
	\end{center}

In this setting \cref{thm:dg} becomes the following.
\begin{thm}
	$\catname{SHGraph}$ is adhesive with respect to the classes
	\begin{gather*}
	\{((h_1,h_2), k, l)\in \arr{SHGraph}\mid (h_1,h_2)\in \reg{\dgr}, k, l \in \mono{Set}\}\\
	\{((h_1,h_2), k, l)\in \arr{SHGraph}\mid (h_1,h_2)\in \mono{\dgr}\}
	\end{gather*}
	while $\catname{DAGHGraph}$ is adhesive with respect to the classes	\begin{gather*}
	\{((h_1,h_2), k, l)\in \arr{DAGHGraph}\mid (h_1,h_2)\in \rtd, k, l \in \mono{Set}\}\\
	\{((h_1,h_2), k, l)\in \arr{DAGHGraph}\mid (h_1,h_2)\in \mono{\dg}\}
	\end{gather*}
\end{thm}

\begin{figure}[t]
	
	\begin{center}
		\begin{tikzpicture}[scale=0.8]
		\node[circle,fill=black,inner sep=0pt,minimum size=6pt,label=left:{$x$}] (V) at (0,0) {};
		
		\node[circle,fill=black,inner sep=0pt,minimum size=6pt,label=right:{$y$}] (W) at (-2,0) {};
		
		\node at(0,2.95){$b$};
		\node at(0,1.65){$a$};	
		\node at(-2,1.65){$c$};	
		
		\draw[-](V)..controls(0.5,0.5)and(0.5,1.25)..(0.25,1.25);
		\draw[-](V)..controls(1,0.5)and(1,2)..(0.25,2.5);
		\draw[-](V)..controls(-0.5,0.5)and(-1.5,1)..(-1.75,1.25);
		\draw[-](W)..controls(-3,0.5)and(-2.5,1.25)..(-2.25,1.25);
		\draw[-latex, red, line width=1pt](-0.5,1.25)--(-1.5,1.25);
		\draw[-latex, red, line width=1pt](-0.5,2.5)..controls(-0.8,2.5)and(-1.5,2.3)..(-1.75,1.8);
		\draw[rounded corners] (-0.25, 1) rectangle (0.25, 1.5) {};
		\draw[rounded corners] (-1.75, 1) rectangle (-2.25, 1.5) {};
		\draw[rounded corners] (-0.25, 2.25) rectangle (0.25, 2.75) {};

		\node[circle,fill=black,inner sep=0pt,minimum size=6pt,label=left:{$x$}] (V) at (5,0) {};
		
		\node[circle,fill=black,inner sep=0pt,minimum size=6pt,label=right:{$y$}] (W) at (3,0) {};

		\node at(5,1.65){$a$};	
		\node at(3,1.75){$b$};	
		
		\draw[-](V)..controls(5.5,0.5)and(5.5,1.25)..(5.25,1.25);
		\draw[-](V)..controls(4.5,0.5)and(3.5,1)..(3.25,1.25);
		\draw[-](W)..controls(2,0.5)and(2.5,1.25)..(2.75,1.25);
		\draw[-latex, red, line width=1pt](4.5,1.25)--(3.5,1.25);
		\draw[latex-, red, line width=1pt](4.5,1.6)..controls(4.25, 1.8)and(3.75,1.8)..(3.5,1.5);
		\draw[rounded corners] (4.75, 1) rectangle (5.25, 1.5) {};
		\draw[rounded corners] (3.25, 1) rectangle (2.75, 1.5) {};
		\end{tikzpicture}
	\end{center}
	\caption{A $\dg$-hypergraph (left) and a $\dgr$-hypergraph corresponding to the CCS process  $P = a(x).b(xy).P$ (right). A red arrow between two edges denotes the order relation $\leq$.}\label{fig:dgdag}
\end{figure}
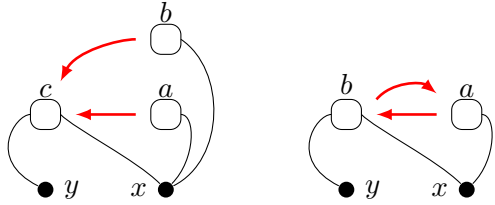

\section{Term graphs} In the past years, the use of a particular class of hypergraphs, called 
\emph{term graphs} has been advocated as a tool for the optimal implementation of terms, with the intuition that 
the graphical counterpart of trees can allow for the sharing of sub-terms~\cite{Plump-termgraph}. A brute force proof of quasiadhesivity of the category of term graphs was given in~\cite{corradini2005term}. In this section we will present the category of term graphs as a subcategory of labeled hypergraphs, moreover we will recover the result of~\cite{corradini2005term} exploiting our new criterion for adhesivity.

\begin{defi}
Let $\Sigma$ be an algebraic signature, a labelled hypergraph $l:\mathcal{G}\to \mathcal{G}^{\Sigma}$ is a \emph{term graph} if for every hyperedges $h_1, h_2\in E_{\mathcal{G}}$, if $t_{\mathcal{G}}(h_1)=t_\mathcal{G}(h_2)$ then $h_1=h_2$. We define $\tg$ to be the full subcategory of $\hyp_{\Sigma}$ and denote by $I_\Sigma$ the corresponding inclusion.
\end{defi}

\begin{rem}Notice that, by \cref{rem:label}, if $\mathcal{G}$ is a term graph then $t_{\mathcal{G}}(h)$ is a word of length $1$, i.e. an element of $V_{\mathcal{G}}$. 
\end{rem}

\begin{exa}Of the examples of \cref{sub:hyper}, only \cref{lab_2} is a term graph.
\end{exa}

 Composing $I_\Sigma$ qith $U_\Sigma:\hyp_{\Sigma}\to \catname{Set}$ we get a functor $U_{\tg}:\tg\to \catname{Set}$. Now, $\Delta_{\Sigma}(X)$ is a term graph for every set $X$, thus $\Delta_{\Sigma}$ factors through $I_\Sigma$. This allows us to conclude the following.
\begin{prop}\label{term:left}The forgetful functor $U_{\tg}:\tg\to \catname{Set}$ has a left adjoint $\Delta_{\tg}$.
 \end{prop}

Take now a mono $(i,j):\mathcal{H}\to \mathcal{G}$ between  $l:\mathcal{G}\to \mathcal{G}^{\Sigma}$ and $l':\mathcal{H}\to \mathcal{G}^{\Sigma}$ in $\hyp$, using \cref{prop:mono}, if $l$ is a term graph then $l'$ belongs to $\tg$ too. In particular we can apply this argument when $l'$ is the equalizer of two parallel arrows between term graphs.
\begin{prop}\label{prop:equ}$\tg$ has equalizers and $I_{\Sigma}$ creates them.
\end{prop}
We have a similar result even for binary products.
\begin{prop} $\tg$ has binary products and $I_\Sigma$ creates them.
\end{prop}
\begin{proof}
	Let $l:\mathcal{G}\to \mathcal{G}^{\Sigma}$ and $t:\mathcal{H}\to \mathcal{G}^{\Sigma}$ be two term graphs, their product in $\hyp_{\Sigma}$ is given by $p:\mathcal{P}\to \mathcal{G}^{\Sigma}$, where the square 
	\begin{center}
		\begin{tikzpicture}
		\node(A) at(2,0){$\mathcal{P}$};
		\node(B) at (4,0){$\mathcal{G}$};
		\node(C) at(4,-1.5){$\mathcal{G}^\Sigma$};
		\node(D) at (2,-1.5){$\mathcal{H}$};
		\draw[->](A)--(B)node[pos=0.5, above]{$(p_E, p_V)$};
		\draw[->](D)--(C)node[pos=0.5, below]{$(l', !_{V_{\mathcal{H}}})$};
		\draw[->](A)--(D)node[pos=0.5, left]{$(q_E, q_V)$};
		\draw[->](B)--(C)node[pos=0.5, right]{$(l, !_{V_{\mathcal{G}}})$};
		\end{tikzpicture}
	\end{center}
	is a pullback in $\hyp$ and $(p, !_{V_\mathcal{P}})$ is the unique diagonal filling it. Since $\hyp$ is a comma category, this means that the squares 
	\begin{center}
		\begin{tikzpicture}
		\node(A) at(2,0){$E_{\mathcal{P}}$};
		\node(B) at (3.5,0){$E_{\mathcal{G}}$};
		\node(C) at(3.5,-1.5){$E_{\mathcal{G}^{\Sigma}}$};
		\node(D) at (2,-1.5){$E_{\mathcal{H}}$};
		\draw[->](A)--(B)node[pos=0.5, above]{$p_E$};
		\draw[->](D)--(C)node[pos=0.5, below]{$l'$};
		\draw[->](A)--(D)node[pos=0.5, left]{$q_E$};
		\draw[->](B)--(C)node[pos=0.5, right]{$l$};
		\node(A) at(5,0){$V_{\mathcal{P}}$};
		\node(B) at (6.5,0){$V_{\mathcal{G}}$};
		\node(C) at(6.5,-1.5){$\{v\}$};
		\node(D) at (5,-1.5){$V_{\mathcal{H}}$};
		\draw[->](A)--(B)node[pos=0.5, above]{$p_V$};
		\draw[->](D)--(C)node[pos=0.5, below]{$!_{V_{\mathcal{H}}}$};
		\draw[->](A)--(D)node[pos=0.5, left]{$q_V$};
		\draw[->](B)--(C)node[pos=0.5, right]{$!_{V_{\mathcal{G}}}$};
		\end{tikzpicture}
	\end{center}
	are pullbacks in $\catname{Set}$. Moreover $t_{\mathcal{P}}$ is such that the diagram
	\begin{center}
	\begin{tikzpicture}
	\node(A) at(2,0){$E_{\mathcal{P}}$};
	\node(B) at (3.5,0){$E_{\mathcal{G}}$};
	\node(C) at(3.5,-1.5){$E_{\mathcal{G}^{\Sigma}}$};
	\node(D) at (2,-1.5){$E_{\mathcal{H}}$};
	\draw[->](A)--(B)node[pos=0.5, above]{$p_E$};
	\draw[->](D)--(C)node[pos=0.5, below]{$l'$};
	\draw[->](A)--(D)node[pos=0.5, left]{$q_E$};
	\draw[->](B)--(C)node[pos=0.5, right]{$l$};
	\node(E) at(0.5,1.5){$V_{\mathcal{P}}^\star$};
	\node(F) at (5,1.5){$V_{\mathcal{G}}^\star$};
	\node(G) at(5,-3){$\{v\}^\star$};
	\node(H) at (0.5,-3){$V_{\mathcal{H}}^\star$};
	\draw[->](E)--(F)node[pos=0.5, above]{$p^\star_V$};
	\draw[->](H)--(G)node[pos=0.5, below]{$\qty(!_{V_{\mathcal{H}}})^\star$};
	\draw[->](E)--(H)node[pos=0.5, left]{$q^\star_V$};
	\draw[->](F)--(G)node[pos=0.5, right]{$\qty(!_{V_{\mathcal{G}}})^\star$};
	\draw[->](A)--(E)node[pos=0.5, right]{$t_{\mathcal{P}}$};
		\draw[->](B)--(F)node[pos=0.5, left]{$t_{\mathcal{G}}$};
	\draw[->](D)--(H)node[pos=0.5, right, xshift=0.1cm]{$t_{\mathcal{H}}$};
	\draw[->](C)--(G)node[pos=0.5, left]{$t_{\mathcal{G}^{\Sigma}}$};	
	\end{tikzpicture}
\end{center}	
	
	is commutative. Take now $h_1, h_{2}\in  E_{\mathcal{P}}$ such that $t_{\mathcal{P}}(h_1)=t_{\mathcal{P}}(h_2)$, then 
	\begin{align*}
t_{\mathcal{G}}(p_E(h_1))&=p^\star_V(t_{\mathcal{P}}(h_1))=p^\star_V(t_{\mathcal{P}}(h_2))=t_{\mathcal{G}}(p_E(h_2))\\
t_{\mathcal{H}}(q_E(h_1))&=q^\star_V(t_{\mathcal{P}}(h_1))=q^\star_V(t_{\mathcal{P}}(h_2))=t_{\mathcal{H}}(q_E(h_2))
	\end{align*}
and thus $p_E(h_1)=p_E(h_2)$ and $q_E(h_1)=q_E(h_2)$, which implies $h_1=h_2$.	
\end{proof}
Since pullbacks can be computed from products and equalizers we also get the following.
\begin{cor}\label{cor:pb}$\tg$ has pullbacks and they are created by $I_\Sigma$.
\end{cor}
\begin{rem}
$\tg$ in general does not have terminal objects. Since $U_{\tg}$ preserves limits, if a terminal object exists it must have the singleton as set of nodes, therefore the set of hyperedges must be empty or a singleton $\{h\}$. Now take as signature the one given by two operations $\{a , b\}$ of arity $0$; we have three term graphs with only one node $v$:
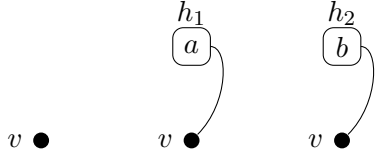
\begin{center}\begin{tikzpicture}
		
		\node[circle,fill=black,inner sep=0pt,minimum size=6pt,label=left:{$v$}] (V) at (5,0) {};
		\node[circle,fill=black,inner sep=0pt,minimum size=6pt,label=left:{$v$}] (U) at (3,0) {};
		\node at(5,1.25){$a$};	
		\node at(5,1.7){$h_1$};	
		
		\draw[-](V)..controls(5.5,0.5)and(5.5,1.25)..(5.25,1.25);
		
		\draw[rounded corners] (4.75, 1) rectangle (5.25, 1.5) {};

		\node[circle,fill=black,inner sep=0pt,minimum size=6pt,label=left:{$v$}] (V) at (7,0) {};
		
		\node at(7,1.25){$b$};	
		\node at(7,1.7){$h_2$};	
		
		\draw[-](V)..controls(7.5,0.5)and(7.5,1.25)..(7.25,1.25);
		
		\draw[rounded corners] (6.75, 1) rectangle (7.25, 1.5) {};
		
	\end{tikzpicture}
\end{center}
 There are no morphisms in $\tg$ between the last two and from the last two to the first one, therefore none of them can be terminal.
\end{rem}

\begin{rem}
	$\tg$ is not an adhesive category. In particular it does not have pushouts along all monomorphisms. Take the graphs of the previous remark and call them $\Delta_{\tg{\Sigma}}(\{v\})$, $l:\mathcal{G}_1\to \mathcal{G}^{\Sigma}$ and $l':\mathcal{G}_2\to \mathcal{G}^{\Sigma}$. The identity $\{v\}\to \{v\}$ induces 
	a span
	
	
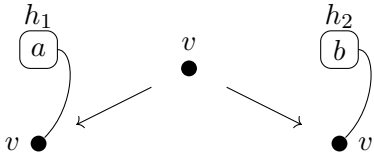
\begin{center}\begin{tikzpicture}
		
		\node[circle,fill=black,inner sep=0pt,minimum size=6pt,label=above:{$v$}] (W) at (5,0) {};
		\node[circle,fill=black,inner sep=0pt,minimum size=6pt,label=left:{$v$}] (U) at (3,-1) {};
		\node at(3,0.25){$a$};	
		\node at(3,0.7){$h_1$};	
		
		\draw[-](U)..controls(3.5,-0.5)and(3.5,0.25)..(3.25,0.25);
		
		\draw[rounded corners] (2.75, 0) rectangle (3.25, 0.5) {};

		\node[circle,fill=black,inner sep=0pt,minimum size=6pt,label=right:{$v$}] (V) at (7,-1) {};
		
		\node at(7,0.25){$b$};	
		\node at(7,0.7){$h_2$};	
		
		\draw[-](V)..controls(7.5,-0.5)and(7.5,0.25)..(7.25,0.25);
		
		\draw[rounded corners] (6.75, 0) rectangle (7.25, 0.5) {};
		\draw[->] (4.5,-0.25)--(3.5, -0.75);
		
		\draw[->] (5.5,-0.25)--(6.5, -0.75);
		\end{tikzpicture}
	\end{center}
	which cannot be completed to any square. Indeed if $h:\mathcal{H}\to \mathcal{G}^\Sigma$ another term graph with $(g_E, g_V):l\to h$ and $(k_E, k_V):l'\to h$  complete the span, than $g_E(h_1)$ and $k_E(h_2)$ both have $g_V(v)=k_V(v)$ has target, thus  $g_E(h_1)=k_E(h_2)$, which implies
	\[a=l(h_1)=h(g_E(h_1))=h(k_E(h_2))=l'(h_2)=b\]
\end{rem}

\begin{defi}
Given a hypergraph  $\mathcal{G}$, we will say that $v\in V_{\mathcal{G}}$ is an \emph{input node} if it does not belong to the image of $t_{\mathcal{G}}$. 
\end{defi}

\begin{prop}\label{prop:image} 
Let $l:\mathcal{H}\to \mathcal{G}^{\Sigma}$ be a term graph and $(f,g):\mathcal{G}\to \mathcal{H}$ an arrow of $\hyp$ such that the image of any input node is still an input node. For every $h\in E_{\mathcal{H}}$, if $t_{\mathcal{H}}(h)\in g(V_{\mathcal{G}})$ then $h\in f(E_{\mathcal{G}})$.
\end{prop}
\begin{proof}
	Let $w\in V_{\mathcal{G}}$ such that $g(w)=t_{\mathcal{H}}(h)$, since $(f,g)$ sends input nodes to input nodes, there exists a $k\in E_{\mathcal{G}}$ such that $t_{\mathcal{G}}(k)=w$. Now,
	\[t_{\mathcal{H}}(f(k))=g(t_{\mathcal{G}}(k))=g(w)=t_{\mathcal{H}}(h)\]
	Since $\mathcal{H}$ is a term graph we can conclude that $f(k)=h$.
\end{proof}

We are now ready to show that regular monos are exactly monos sending input nodes to input nodes.

\begin{lem}\label{lem:reg} A mono $(i,j)$ between two term graphs  $l:\mathcal{G}\to \mathcal{G}^{\Sigma}$ and  $l':\mathcal{H}\to \mathcal{G}^{\Sigma}$ is regular if and only if it sends input nodes to input nodes.
\end{lem}

\begin{proof}
	$(\Rightarrow)$. This follows at once from \cref{prop:equ}.
		
	\noindent	
	$(\Leftarrow)$. Take $V$ and $E$ to be, respectively, $V_{\mathcal{H}}\sqcup( V_{\mathcal{H}}\smallsetminus j(V_{\mathcal{G}}))$ and $E_{\mathcal{H}}\sqcup ( E_{\mathcal{H}}\smallsetminus i(E_{\mathcal{G}}))$, with inclusions
		\[ j_1:V_{\mathcal{H}}\to V\qquad j_2:V_{\mathcal{H}}\smallsetminus j(V_{\mathcal{G}})\to V \qquad  i_1:E_{\mathcal{H}}\to E \qquad i_2:E_{\mathcal{H}}\smallsetminus i(E_{\mathcal{G}})\to E \]
		Now, we are going to use another auxiliary function 
		\[r:V_{\mathcal{H}}\to V \qquad v\mapsto \begin{cases}
		j_1(v) & v \in j(V_{\mathcal{G}})\\
		j_2(v) & v\notin j(V_{\mathcal{G}})
		\end{cases}\]
		which is clearly injective.
		
		So equipped we can define $s,t:E\rightrightarrows V^\star$ as the functions induced by
		\begin{equation*}
		\begin{split}
	s_1&: E_{\mathcal{H}}\to V^\star  \quad h\mapsto j^\star_1\qty(s_{\mathcal{H}}\qty(h))\\
		s_2&:E_{\mathcal{H}}\smallsetminus i(E_{\mathcal{G}})\to V^\star \quad h\mapsto r^\star(s_{\mathcal{H}}(h))
		\end{split} \qquad 
		\begin{split}t_1&: E_{\mathcal{H}}\to V^\star  \quad h\mapsto j^\star_1\qty(t_{\mathcal{H}}\qty(h))
		\\ 
		t_2&:E_{\mathcal{H}}\smallsetminus i(E_{\mathcal{G}})\to V^\star \qquad h\mapsto r^\star(t_{\mathcal{H}}(h))
		\end{split}
		\end{equation*}
	
	Let now $\mathcal{K}$ be the hypergraph $(E, V, s, t)$, and take as label $q:\mathcal{K}\to \mathcal{G}^{\Sigma}$ the morphism induced by $l':E_{\mathcal{H}}\to E_{\mathcal{G}^{\Sigma}}$ and its restriction to $E_{\mathcal{H}}\smallsetminus i(E_{\mathcal{G}})$.	We have now to check that  $q:\mathcal{K}\to \mathcal{G}^{\Sigma}$ is actually a term graph. Suppose that $t(h_1)=t(h_2)$, we have three cases.
	\begin{itemize}
		\item $h_1=i_1(h)$ and $h_2=i_1(k)$ for some $h$, $k$ in $E_{\mathcal{H}}$. Then \[j^{\star}_{1}\qty(t_{\mathcal{H}}(h))=t(i_1(h))=t(h_1)=t(h_2)=t(i_1(k))=j^{\star}_{1}\qty(t_{\mathcal{H}}(k))\] 
		But $j^\star_1$ is injective and thus
		\[t_{\mathcal{H}}(h)=t_{\mathcal{H}}(k)\]
		from which the thesis follows since $l':\mathcal{H}\to \mathcal{G}^{\Sigma}$ is a term graph.
		\item $h_1=i_2(h)$ and $h_2=i_2(k)$ for some $h$, $k$ in $E_{\mathcal{H}}\smallsetminus i(E_{\mathcal{G}})$. As before we can compute to get \[r^\star\qty(t_{\mathcal{H}}(h))=t(i_2(h))=t(h_1)=t(h_2)=t(i_2(k))=r^\star\qty(t_{\mathcal{H}}(k))\] 
		and thus, exploiting \cref{rem:mono}, $h_1=h_2$.
		\item $h_1=i_1(h)$ and $h_2=i_2(k)$ for some $h\in E_{\mathcal{H}}$, $k$ in $E_{\mathcal{H}}\smallsetminus i(E_{\mathcal{G}})$. By the definition of $t$, this can happen only if $t_{\mathcal{H}}(k)\in j(V_{\mathcal{G}})$, therefore, using \cref{prop:image}, $k$ must be an element of $i(E_{\mathcal{G}})$, which is absurd.
		\item $h_1=i_2(h)$ and $h_2=i_1(k)$ for some $h\in E_{\mathcal{H}}$, $k$ in $E_{\mathcal{H}}\smallsetminus i(E_{\mathcal{G}})$. This is done as in the previous point, switching the roles of $h_1$ and $h_2$.
	\end{itemize}
	Now, by construction $(i_1, j_1)$ defines an arrow $\mathcal{H}\to \mathcal{K}$, which is also a morphism of $\tg$. On the other hand we can construct another arrow $(f,r)$ parallel to it defining
	\begin{align*}
	f:E_{\mathcal{H}}\to E \quad h \mapsto
	\begin{cases}
	i_1(h) &h\in i(E_{\mathcal{G}})\\
	i_2(h) &h\notin i(E_{\mathcal{G}})
	\end{cases}
	\end{align*}
	and noticing that 
	\begin{align*}s(f(h))&=\begin{cases}
	s_1(h)&h\in i(E_{\mathcal{G}})\\
	s_2(h) &h\notin i(E_{\mathcal{G}})
	\end{cases}=\begin{cases}
	j^\star_1(s_{\mathcal{H}}(h))&h\in i(E_{\mathcal{G}})\\
	r^\star(s_{\mathcal{H}}(h)) &h\notin i(E_{\mathcal{G}})
	\end{cases}=r^\star\qty(s_{\mathcal{H}}\qty(h))\\
	t(f(h))&=\begin{cases}
	t_1(h)&h\in i(E_{\mathcal{G}})\\
	t_2(h) &h\notin i(E_{\mathcal{G}})
	\end{cases}=\begin{cases}
	j^\star_1(t_{\mathcal{H}}(h))&h\in i(E_{\mathcal{G}})\\
	r^\star(t_{\mathcal{H}}(h)) &h\notin i(E_{\mathcal{G}})
	\end{cases}=r^\star\qty(t_{\mathcal{H}}\qty(h))
	\end{align*}
	Where the last equalities follows since $h\in i(E_{\mathcal{G}})$ implies that
	\[s_{\mathcal{H}}(h)=s_{\mathcal{H}}\qty(i\qty(k))=j^\star\qty(s_{\mathcal{G}}(k))\qquad t_{\mathcal{H}}(h)=t_{\mathcal{H}}\qty(i\qty(k))=j^\star\qty(t_{\mathcal{G}}(k))\]

	By construction
	$q(f(h))=l'(h)$, thus $(f,g)$ is a morphism in $\tg$. Now, $\mathcal{G}$ is the equalizer of $(f,g)$ and $(i,j)$ in $\hyp$, thus it is their equalizer even in $\hyp_{\Sigma}$, and the thesis follows since the inclusion $\tg\to \hyp_{\Sigma}$ reflects limits.
\end{proof}

\begin{prop}\label{prop:push}
	Let $l_0:\mathcal{G}\to \mathcal{G}^{\Sigma}$, $l_1:\mathcal{H}\to \mathcal{G}^{\Sigma}$ and $l_2:\mathcal{K}\to \mathcal{G}^{\Sigma}$ be term graphs and $(f_1, g_1):\mathcal{G}\to \mathcal{H}$, $(f_2, g_2):\mathcal{G}\to \mathcal{K}$ two morphisms between them and suppose that $(f_1, g_1)$ is a regular mono. Then their pushout $p:\mathcal{P}\to \mathcal{G}^{\Sigma}$ in $\hyp_{\Sigma}$ is a term graph too.
\end{prop}
\begin{rem}
By definition $\hyp_{\Sigma}$  the comma category on $\id{\hyp}$ and the costant functor in $\mathcal{G}^\Sigma$. Now, this last functor preserves pushouts, thus we know how to compute this kind of colimits in $\hyp_{\Sigma}$. In particular $p:\mathcal{P}\to \mathcal{G}^{\Sigma}$ is given by the pushout $\mathcal{P}$ in $\hyp$ equipped with the labeling induced by $l_1$ and $l_2$.
\end{rem}
\begin{proof}
	By the previous remark we know that we have pushout squares in $\catname{Set}$
	\begin{center}
		\begin{tikzpicture}
		\node(A) at(0,0){$E_{\mathcal{G}}$};
		\node(B) at (1.5,0){$E_{\mathcal{K}}$};
		\node(C) at(1.5,-1.5){$E_{\mathcal{P}}$};
		\node(D) at (0,-1.5){$E_{\mathcal{H}}$};
		\draw[->](A)--(B)node[pos=0.5, above]{$g_E$};
		\draw[->](D)--(C)node[pos=0.5, below]{$h_E$};
		\draw[->](A)--(D)node[pos=0.5, left]{$f_E$};
		\draw[->](B)--(C)node[pos=0.5, right]{$k_E$};

		\node(A) at(3,0){$V_{\mathcal{G}}$};
		\node(B) at (4.5,0){$V_{\mathcal{K}}$};
		\node(C) at(4.5,-1.5){$V_{\mathcal{P}}$};
		\node(D) at (3,-1.5){$V_{\mathcal{H}}$};
		\draw[->](A)--(B)node[pos=0.5, above]{$g_V$};
		\draw[->](D)--(C)node[pos=0.5, below]{$h_V$};
		\draw[->](A)--(D)node[pos=0.5, left]{$f_V$};
		\draw[->](B)--(C)node[pos=0.5, right]{$k_V$};
		\end{tikzpicture}
	\end{center} 
	And diagrams
	
		
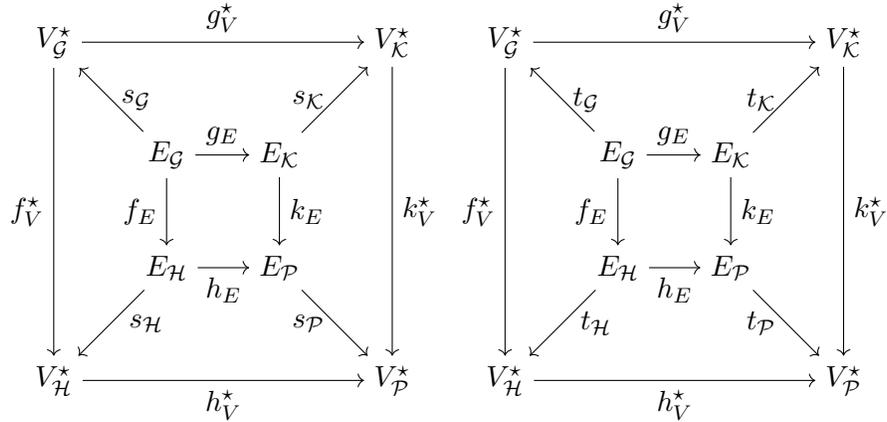
\begin{center}
		\begin{tikzpicture}
		\node(A) at(-4,0){$E_{\mathcal{G}}$};
		\node(B) at (-2.5,0){$E_{\mathcal{K}}$};
		\node(C) at(-2.5,-1.5){$E_{\mathcal{P}}$};
		\node(D) at (-4,-1.5){$E_{\mathcal{H}}$};
		\draw[->](A)--(B)node[pos=0.5, above]{$g_E$};
		\draw[->](D)--(C)node[pos=0.5, below]{$h_E$};
		\draw[->](A)--(D)node[pos=0.5, left]{$f_E$};
		\draw[->](B)--(C)node[pos=0.5, right]{$k_E$};
		\node(E) at(-5.5,1.5){$V_{\mathcal{G}}^\star$};
		\node(F) at (-1,1.5){$V_{\mathcal{K}}^\star$};
		\node(G) at(-1,-3){$V_{\mathcal{P}}^\star$};
		\node(H) at (-5.5,-3){$V_{\mathcal{H}}^\star$};
		\draw[->](E)--(F)node[pos=0.5, above]{$g^\star_V$};
		\draw[->](H)--(G)node[pos=0.5, below]{$h^\star_V$};
		\draw[->](E)--(H)node[pos=0.5, left]{$f^\star_V$};
		\draw[->](F)--(G)node[pos=0.5, right]{$k^\star_V$};
		\draw[->](A)--(E)node[pos=0.5, right]{$s_{\mathcal{G}}$};
		\draw[->](B)--(F)node[pos=0.5, left]{$s_{\mathcal{K}}$};
		\draw[->](D)--(H)node[pos=0.5, right, xshift=0.1cm]{$s_{\mathcal{H}}$};
		\draw[->](C)--(G)node[pos=0.5, left]{$s_{\mathcal{P}}$};

		\node(A') at(2,0){$E_{\mathcal{G}}$};
		\node(B') at (3.5,0){$E_{\mathcal{K}}$};
		\node(C') at(3.5,-1.5){$E_{\mathcal{P}}$};
		\node(D') at (2,-1.5){$E_{\mathcal{H}}$};
		\draw[->](A')--(B')node[pos=0.5, above]{$g_E$};
		\draw[->](D')--(C')node[pos=0.5, below]{$h_E$};
		\draw[->](A')--(D')node[pos=0.5, left]{$f_E$};
		\draw[->](B')--(C')node[pos=0.5, right]{$k_E$};
		\node(E') at(0.5,1.5){$V_{\mathcal{G}}^\star$};
		\node(F') at (5,1.5){$V_{\mathcal{K}}^\star$};
		\node(G') at(5,-3){$V_{\mathcal{P}}^\star$};
		\node(H') at (0.5,-3){$V_{\mathcal{H}}^\star$};
		\draw[->](E')--(F')node[pos=0.5, above]{$g^\star_V$};
		\draw[->](H')--(G')node[pos=0.5, below]{$h^\star_V$};
		\draw[->](E')--(H')node[pos=0.5, left]{$f^\star_V$};
		\draw[->](F')--(G')node[pos=0.5, right]{$k^\star_V$};
		\draw[->](A')--(E')node[pos=0.5, right]{$t_{\mathcal{G}}$};
		\draw[->](B')--(F')node[pos=0.5, left]{$t_{\mathcal{K}}$};
		\draw[->](D')--(H')node[pos=0.5, right, xshift=0.1cm]{$t_{\mathcal{H}}$};
		\draw[->](C')--(G')node[pos=0.5, left]{$t_{\mathcal{P}}$};	
		\end{tikzpicture}
	\end{center}	
	 Now, suppose that there exists $h_1, h_2\in E_{\mathcal{P}}$ such that $t_{\mathcal{P}}(h_1)=t_{\mathcal{P}}(h_2)$, by \cref{rem:label}  we know that $t_{\mathcal{P}}(h_1), t_{\mathcal{P}}(h_2)\in V_{\mathcal{P}}$, thus, by \cref{lem:push} there are two possible cases.
	 \begin{enumerate}[label=(\alph*)]
	 	\item There exists a unique $w\in V_{\mathcal{K}}$ such that $t_{\mathcal{P}}(h_1)=k_V(w)=t_{\mathcal{P}}(h_2)$. Using again \cref{lem:push} we can split this case in four subcases.
	 	\medskip 
	 	
	 	\noindent (a.i) There exist $k_1$ and $k_2\in E_{\mathcal{K}}$ such that 
	 	\[k_E(k_1)=h_1 \qquad k_E(k_2)=h_2\]
	 	Then uniqueness of $w$ implies that
	 	\[t_{\mathcal{K}}(k_1)=w=t_{\mathcal{K}}(k_2)\]
	 	and we can conclude since $\mathcal{K}$ is the hypergraph undelying a term graph.
	 	\medskip 
	 	
	 	\noindent (a.ii) There exist $h'_1$ and $h'_2\in E_{\mathcal{H}}$ such that 
	 	\[h_E(h'_1)=h_1 \qquad h_E(h'_2)=h_2\]
	 	Therefore
	 	\[h_V(t_{\mathcal{H}}(h'_1))=h_V(t_{\mathcal{H}}(h'_2))\]
	 	hence there exist $v_1, v_2\in V_{\mathcal{G}}$ such that
	 	\[g_V(v_1)=g^\star_V(v_2) \qquad k_V(g_V(v_1))=t_{\mathcal{P}}(h_1)=t_{\mathcal{P}}(h_2)=k_V(g_V(v_2))\]
	 	Thus 
	 	\[g_V(v_1)=w=g_V(v_2)\]
	 	On the other hand \cref{prop:image} implies that there exist $g_1, g_2\in E_{\mathcal{G}}$ such that
	 	\[h'_1=f_E(g_1) \qquad h'_2=f_E(g_2) \qquad t_{\mathcal{G}}(g_1)=v_1 \qquad t_{\mathcal{G}}(g_2)=v_2\]	
	 Hence
	 \[t_{\mathcal{K}}(g_E(g_1))=w=t_{\mathcal{K}}(g_E(g_2))\]
	 and we can deduce that 	
	 \[g_E(g_1)=g_E(g_2)\]
	from which $h_1=h_2$ follows using \cref{lem:push}.
	 	 \medskip
	 	 
	 	\noindent (a.iii) There exist $k\in E_{\mathcal{K}}$ and $h'\in E_{\mathcal{H}}$ such that
	 	\[h_E(h')=h_1\qquad k_E(k)=k_2\]
	 	hence
	 	\[h_V(t_{\mathcal{H}}(h'))=k_V(t_{\mathcal{K}}(k))\]
	 	and we can conclude that there exists $v\in V_{\mathcal{G}}$ with the property that
	 	\[f_V(v)=t_{\mathcal{H}}(h')\qquad g_V(v)=w\] 
	 	Using \cref{prop:image} we can also deduce the existence of $g\in E_\mathcal{G}$ satisfying 
	 	\[ f_E(g)=h' \qquad t_{\mathcal{G}}(g)=v\]
	 	Since $\mathcal{K}$ underlies a term graph it follows that $g_E(g)=k$. Appealing again to \cref{lem:push} we get the thesis.
	 	\medskip 
	 	
	 	\noindent (a.iv) This is case is dealt as the previous one, simply swapping $h_1$ and $h_2$.
	 	\item  There exists a unique $v\in V_\mathcal{H}\smallsetminus f_V(V_{\mathcal{G}})$ such that \[t_{\mathcal{P}}(h_1)=h_V(v)=t_{\mathcal{P}}(h_2)\]
	 	Now, if $h\in E_{\mathcal{K}}$ is such that $t_{\mathcal{P}}(k_E(h))=h_V(v)$ then, by \cref{lem:push}, there must be $w\in V_{\mathcal{G}}$ such that $f_V(w)=v$ and $g_V(w)=k$, but this is absurd under our hypothesis. We conclude that there  exist  $h'_1$ and $h'_2\in E_{\mathcal{H}}$ such that 
	 	\[h_E(h'_1)=h_1 \qquad h_E(h'_2)=h_2\]
	 and the uniqueness of $v$ implies that 
	 \[t_{\mathcal{H}}(h'_1)=t_{\mathcal{H}}(h'_2)\]
	 The thesis now follows.	 	 \qedhere 
	 \end{enumerate}
\end{proof}
\cref{cor:varie3,lem:reg,prop:push} allow us to recover the following result, previously proved by direct computation in \cite[Thm.~4.2]{corradini2005term}.
\begin{cor}
	The category $\tg$ is quasi-adhesive.
\end{cor}

\section{Conclusions}\label{sec:concl}
In this paper we have introduced a new criterion for $\mathcal{M}, \mathcal{N}$-adhesivity, based on the verification of some properties of functors connecting the category of interest to a family of suitably adhesive categories.  This criterion can be seen as a distilled abstraction of many \emph{ad hoc} proofs of adhesivity found in literature.
This criterion allows us to prove in a uniform and systematic way some previous results about the adhesivity of categories built by products, exponents, and comma construction. We have applied the criterion to several significant examples, such as term graphs and directed (acyclic) graphs; moreover, using the modularity of our approach, we have readily proved suitable adhesivity properties to categories constructed by combining simpler ones.  In particular, we have been able to tackle the adhesivity problem for several categories of hierarchical (hyper)graphs, including Milner's bigraphs, bigraphs with sharing, and a new version of bigraphs with recursion.

As future work, we plan to analyse other categories of graph-like objects using our criterion; an interesting case is that of \emph{directed bigraphs} \cite{gm:directedbig,bgm:calco09,bmp:sac}.
Moreover, it is worth to verify whether the $\mathcal{M}, \mathcal{N}$-adhesivity that we obtain from the results of this paper is suited for modelling specific rewriting systems, e.g.~based on the DPO approach.  As an example, $\tg$ is quasiadhesive but this does not suffice in most applications, because the rules are often spans of monomorphisms, and not of regular monos \cite{corradini2005term}. 

\bibliographystyle{alphaurl}
\bibliography{bibliog.bib}

\newcommand{\etalchar}[1]{$^{#1}$}
\begin{thebibliography}{EHKPP91}

\bibitem[ACR19]{azzi2019essence}
Guilherme~Grochau Azzi, Andrea Corradini, and Leila Ribeiro.
\newblock On the essence and initiality of conflicts in
  {$\mathcal{M}$}-adhesive transformation systems.
\newblock {\em Journal of Logical and Algebraic Methods in Programming},
  109:100482, 2019.

\bibitem[AHS06]{adamek2004abstract}
Ji{\v{r}}{\'\i} Ad{\'a}mek, Horst Herrlich, and George~E. Strecker.
\newblock Abstract and concrete categories: The joy of cats.
\newblock {\em Reprints in Theory and Applications of Categories}, 17:1--507,
  2006.

\bibitem[BGM09]{bgm:calco09}
Giorgio Bacci, Davide Grohmann, and Marino Miculan.
\newblock {DBtk}: {A} toolkit for directed bigraphs.
\newblock In Alexander Kurz, Marina Lenisa, and Andrzej Tarlecki, editors, {\em
  {CALCO} 2009}, volume 5728 of {\em LNCS}, pages 413--422. Springer, 2009.

\bibitem[BMP20]{bmp:sac}
Fabio Burco, Marino Miculan, and Marco Peressotti.
\newblock Towards a formal model for composable container systems.
\newblock In Chih{-}Cheng Hung, Tom{\'{a}}s Cern{\'{y}}, Dongwan Shin, and
  Alessio Bechini, editors, {\em SAC 2020}, pages 173--175. {ACM}, 2020.

\bibitem[CG05]{corradini2005term}
Andrea Corradini and Fabio Gadducci.
\newblock On term graphs as an adhesive category.
\newblock In Maribel Fern\'andez, editor, {\em TERMGRAPH 2004}, volume 127(5)
  of {\em ENTCS}, pages 43--56. Elsevier, 2005.

\bibitem[CJ95]{carboni1995connected}
Aurelio Carboni and Peter Johnstone.
\newblock Connected limits, familial representability and {A}rtin glueing.
\newblock {\em Mathematical Structures in Computer Science}, 5(4):441--459,
  1995.

\bibitem[CMR{\etalchar{+}}97]{CorradiniMREHL97}
Andrea Corradini, Ugo Montanari, Francesca Rossi, Hartmut Ehrig, Reiko Heckel,
  and Michael L{\"{o}}we.
\newblock Algebraic approaches to graph transformation - {P}art {I}: {B}asic
  concepts and double pushout approach.
\newblock In Grzegorz Rozenberg, editor, {\em Handbook of Graph Grammars and
  Computing by Graph Transformations, Volume 1: Foundations}, pages 163--246.
  World Scientific, 1997.

\bibitem[EEPT06]{ehrig2006fundamentals}
Hartmut Ehrig, Karsten Ehrig, Ulrike Prange, and Gabriele Taentzer.
\newblock {\em Fundamentals of Algebraic Graph Transformation}.
\newblock Springer, 2006.

\bibitem[EHKPP91]{ehrig1991parallelism}
Hartmut Ehrig, Annegret Habel, Hans-J{\"o}rg Kreowski, and Francesco
  Parisi-Presicce.
\newblock Parallelism and concurrency in high-level replacement systems.
\newblock {\em Mathematical Structures in Computer Science}, 1(3):361--404,
  1991.

\bibitem[EHPP04]{ehrig2004adhesive}
Hartmut Ehrig, Annegret Habel, Julia Padberg, and Ulrike Prange.
\newblock Adhesive high-level replacement categories and systems.
\newblock In Hartmut Ehrig, Gregor Engels, Francesco Parisi{-}Presicce, and
  Grzegorz Rozenberg, editors, {\em ICGT 2004}, LNCS, pages 144--160. Springer,
  2004.

\bibitem[GL12]{garner2012axioms}
Richard Garner and Stephen Lack.
\newblock On the axioms for adhesive and quasiadhesive categories.
\newblock {\em Theory and Applications of Categories}, 27(3):27--46, 2012.

\bibitem[GM07]{gm:directedbig}
Davide Grohmann and Marino Miculan.
\newblock Directed bigraphs.
\newblock In Marcelo Fiore, editor, {\em MFPS 2007}, volume 173 of {\em ENTCS},
  pages 121--137. Elsevier, 2007.

\bibitem[HP12]{habel2012mathcal}
Annegret Habel and Detlef Plump.
\newblock $\mathcal{M}$, $\mathcal{N}$-adhesive transformation systems.
\newblock In Hartmut Ehrig, Gregor Engels, Hans{-}J{\"{o}}rg Kreowski, and
  Grzegorz Rozenberg, editors, {\em ICGT 2012}, volume 7562 of {\em LNCS},
  pages 218--233. Springer, 2012.

\bibitem[JLS07]{johnstone2007quasitoposes}
Peter~T. Johnstone, Stephen Lack, and Pawel Sobocinski.
\newblock Quasitoposes, quasiadhesive categories and {A}rtin glueing.
\newblock In Till Mossakowski, Ugo Montanari, and Magne Haveraaen, editors,
  {\em CALCO 2007}, volume 4624 of {\em LNCS}, pages 312--326. Springer, 2007.

\bibitem[Lei04]{leinster2004higher}
Tom Leinster.
\newblock {\em Higher operads, higher categories}.
\newblock Cambridge University Press, 2004.

\bibitem[LS05]{lack2005adhesive}
Stephen Lack and Pawe{\l} Soboci{\'n}ski.
\newblock Adhesive and quasiadhesive categories.
\newblock {\em RAIRO-Theoretical Informatics and Applications}, 39(3):511--545,
  2005.

\bibitem[LS06]{lack2006toposes}
Stephen Lack and Pawel Sobocinski.
\newblock Toposes are adhesive.
\newblock In Andrea Corradini, Hartmut Ehrig, Ugo Montanari, Leila Ribeiro, and
  Grzegorz Rozenberg, editors, {\em ICGT 2006}, volume 4178 of {\em LNCS},
  pages 184--198. Springer, 2006.

\bibitem[Mil09]{milner:bigraphs}
Robin Milner.
\newblock {\em The Space and Motion of Communicating Agents}.
\newblock Cambridge University Press, 2009.

\bibitem[ML13]{mac2013categories}
Saunders Mac~Lane.
\newblock {\em Categories for the working mathematician}.
\newblock Springer, 2013.

\bibitem[MO12]{nyko2012}
Nikos Mylonakis and Fernando Orejas.
\newblock A framework of hierarchical graphs and its application to the
  semantics of {SRML}.
\newblock Technical Report LSI-12-1-R, Facultad de Inform\'atica, Universitat
  Polit\`ecnica da Catalunya, 2012.
\newblock URL: \url{https://upcommons.upc.edu/handle/2117/91279}.

\bibitem[Pad17]{Padberg17}
Julia Padberg.
\newblock Hierarchical graph transformation revisited - {T}ransformations of
  coalgebraic graphs.
\newblock In Juan de~Lara and Detlef Plump, editors, {\em ICGT 2017}, volume
  10373 of {\em LNCS}, pages 20--35. Springer, 2017.

\bibitem[Pal04]{palacz2004algebraic}
Wojciech Palacz.
\newblock Algebraic hierarchical graph transformation.
\newblock {\em Journal of Computer and System Sciences}, 68(3):497--520, 2004.

\bibitem[PH16]{peuser2016composition}
Christoph Peuser and Annegret Habel.
\newblock Composition of {$\mathcal{M},\mathcal{N}$}-adhesive categories with
  application to attribution of graphs.
\newblock In Detlef Plump, editor, {\em GCM 2015}, volume~73 of {\em Electronic
  Communications of the EASST}. EASST, 2016.

\bibitem[Plu99]{Plump-termgraph}
Detlef Plump.
\newblock Term graph rewriting.
\newblock In H.~Ehrig, G.~Engels, H.-J. Kreowski, and G.~Rozenberg, editors,
  {\em Handbook of Graph Grammars and Computing by Graph Transformations, Vol.
  2: Applications, Languages, and Tools}, pages 3--61. World Scientific, 1999.

\bibitem[SB20]{sobocinski2020rule}
Pawe{\l} Soboci{\'n}ski and Nicolas Behr.
\newblock Rule algebras for adhesive categories.
\newblock {\em Logical Methods in Computer Science}, 16, 2020.

\bibitem[SC15]{sc:bigraphsharing}
Michele Sevegnani and Muffy Calder.
\newblock Bigraphs with sharing.
\newblock {\em Theoretical Computer Science}, 577:43--73, 2015.
\newblock \href {https://doi.org/10.1016/j.tcs.2015.02.011}
  {\path{doi:10.1016/j.tcs.2015.02.011}}.

\end{thebibliography}

\end{document}